\def\buildcombined{1}

\documentclass[onefignum,onetabnum]{siamart251216}

\newif\ifmainpaper
\newif\ifsupplementcontent
\ifdefined\buildsupplement
  \mainpaperfalse
  \supplementcontenttrue
\else
  \mainpapertrue
  \ifdefined\buildcombined
    \supplementcontenttrue
  \else
    \supplementcontentfalse
  \fi
\fi

\usepackage{algorithm}
\usepackage{algorithmic}
\usepackage{enumerate}
\usepackage{booktabs}
\usepackage{bm}
\usepackage{amssymb}
\usepackage{colortbl}
\usepackage{graphicx}
\usepackage{url}
\usepackage{placeins}
\usepackage{tikz}
\usetikzlibrary{positioning,arrows.meta,shapes.geometric,fit}
\ifmainpaper
\hypersetup{pdftitle={An Adjoint Krylov Approach to Parameter Inference for Affine Families of Linear Evolution Equations}}
\else
\hypersetup{pdftitle={Supplementary Material for An Adjoint Krylov Approach to Parameter Inference for Affine Families of Linear Evolution Equations}}
\fi
\hypersetup{
  pdfauthor={A. Dendukuri, S. Chandrasekaran and L. Petzold},
  colorlinks=true, linkcolor=black, citecolor=black, urlcolor=black
}

\newsiamremark{remark}{Remark}
\newsiamremark{example}{Example}
\newsiamthm{assumption}{Assumption}

\newcommand{\R}{\mathbb{R}}
\newcommand{\N}{\mathbb{N}}
\newcommand{\calL}{\mathcal{L}}
\newcommand{\calX}{\mathcal{X}}

\DeclareMathOperator{\diag}{diag}

\newcommand{\bA}{\bm{A}}

\newcommand{\bJ}{\bm{J}}

\newcommand{\bP}{\bm{P}}
\newcommand{\bH}{\bm{H}}

\newcommand{\bE}{\bm{E}}
\newcommand{\bG}{\bm{G}}

\newcommand{\bV}{\bm{V}}

\newcommand{\bD}{\bm{D}}

\newcommand{\blambda}{\bm{\lambda}}
\newcommand{\bp}{\bm{p}}
\newcommand{\br}{\bm{r}}
\newcommand{\bv}{\bm{v}}
\newcommand{\bw}{\bm{w}}
\newcommand{\bx}{\bm{x}}
\newcommand{\by}{\bm{y}}

\newcommand{\bdelta}{\bm{\delta}}
\newcommand{\btheta}{\bm{\theta}}
\newcommand{\bphi}{\bm{\phi}}
\newcommand{\bzeta}{\bm{\zeta}}
\newcommand{\bnu}{\bm{\nu}}

\newcommand{\supp}{\mathrm{supp}}
\newcommand{\nnz}{\mathrm{nnz}}

\ifmainpaper
\headers{Adjoint Krylov Gradients for Parameter Inference}{A. Dendukuri, S. Chandrasekaran, and L. Petzold}
\title{An Adjoint Krylov Approach to Parameter Inference for Affine Families of Linear Evolution Equations}
\else
\headers{Supplement (Parameter Inference)}{A. Dendukuri, S. Chandrasekaran, and L. Petzold}
\title{Supplementary Material for ``An Adjoint Krylov Approach to Parameter Inference for Affine Families of Linear Evolution Equations''}
\fi
\author{Aditya Dendukuri\thanks{Department of Computer Science, University of California, Santa Barbara
  (\email{aditya\_dendukuri@ucsb.edu}, \email{petzold@engineering.ucsb.edu}).}
\and Shivkumar Chandrasekaran\thanks{Department of Electrical and Computer Engineering, University of California, Santa Barbara
  (\email{shiv@ece.ucsb.edu}).}
\and Linda Petzold\footnotemark[2]}

\begin{document}
\maketitle

\ifmainpaper

\begin{abstract}
This paper introduces an adjoint Krylov method for parameter inference in
large, sparse linear evolution equations whose system matrices depend
affinely on unknown coefficients. We estimate these coefficients from
discrete observations by constructing one forward and one adjoint Krylov
approximation on each observation interval. Quadrature of the resulting
bilinear form gives the full coefficient gradient through sparse
matrix--vector products, without evaluating a separate Fr\'{e}chet derivative
in every operator direction. We separate the gradient error into forward,
adjoint, and quadrature terms and derive residual-based bounds for Markov and
Lindblad generators. We also identify cases in which heuristic stopping
estimates become small before the gradient has converged. With workspace
reuse, we compute library gradients for a 12-qubit Lindblad generator with
dimension $1.7\times10^7$ on a single machine with 16~GB of memory.
Experiments on chemical master, Kuramoto--Sivashinsky, and Lindblad equations
demonstrate accurate gradients and parameter recovery.

\end{abstract}
\begin{keywords}
matrix exponential, Fr\'{e}chet derivative, Krylov methods, adjoint methods,
sparse identification, inverse problems
\end{keywords}
\begin{AMS}
65F60, 65F50, 65L09, 65K10
\end{AMS}

\section{Introduction}
\label{sec:intro}

Many evolution models become large systems of linear ordinary differential
equations after state-space truncation or spatial discretization. We consider
models of the form
\begin{equation}
\label{eq:intro_evolution}
 \dot{\bm u}(t)=\bA(\btheta)\bm u(t),\qquad
 \bA(\btheta)=\bA_0+\sum_j\theta_j\bA_j,
 \qquad
 \bm u(t+h)=e^{h\bA(\btheta)}\bm u(t),
\end{equation}
where the matrices $\bA_j$ are known and the coefficients $\theta_j$ are to
be estimated from observations at discrete times. Depending on the
application, $\bm u$ may be a probability distribution, a discretized field,
or a vectorized density matrix, while the directions $\bA_j$ describe
candidate reactions and propensities, differential operators, or Hamiltonian
and dissipative terms. If the candidate library contains more terms than the
model requires, sparse optimization selects a small active set.

A central bottleneck in parameter inference is the evaluation of gradients
through the exponential action in \eqref{eq:intro_evolution}. One can compute
a separate Fr\'{e}chet derivative in every library direction, but this
repeats the same propagation for each coefficient. Forming the derivative
with respect to all entries of $\bA$ avoids this repetition only by
introducing a dense matrix
that is too large to store. For scalar losses, this difficulty can be avoided
with the classical adjoint formula, which expresses every coefficient
derivative in terms of the same forward and adjoint solutions, so that only
the contraction against $\bA_j$ changes across the library.

In this paper we compute the forward and adjoint paths in separate Krylov
spaces and integrate the resulting bilinear form by quadrature. The two
spaces are constructed once for an observation interval and are then reused
for every library direction. At each quadrature node, we contract
the forward and adjoint vectors against $\bA_j$, requiring one sparse
matrix--vector product and one inner product per direction. The adjoint path
also gives the derivative with respect to the state entering the interval,
while differentiating the endpoint with respect to $h$ gives the interval
derivative. We use these derivatives when the exponential action appears
inside a nonlinear time step or follows a parameter map.

The forward Arnoldi approximation, the adjoint Arnoldi approximation, and the
quadrature rule introduce separate errors. We derive a gradient-error
decomposition that retains all three contributions and use the Arnoldi
residuals to bound the first two. For Markov and Lindblad generators,
contractivity yields computable bounds in the natural paired norms. Keeping
the forward and adjoint errors separate also reveals why two commonly used
refinement estimates can become small too soon. A multiplicative estimate
vanishes when either Arnoldi process terminates, even if the other process
remains unresolved, while a difference between successive approximations can
be misleading during slow convergence.

We test the method on three parameter-inference problems. The first recovers
a sparse chemical master equation from distribution snapshots separated by
long observation intervals. The second differentiates a
Kuramoto--Sivashinsky time step in which a nonlinear map precedes the
exponential action. The third estimates Hamiltonian and dissipative terms in
a Lindblad equation with complex arithmetic and large vectorized density
states. Together, these examples test the calculation across stiff,
nonnormal, nonlinear, and complex-valued models.

\subsection{Contributions}
The principal contributions are as follows.
\begin{itemize}
\item We give an adjoint Krylov algorithm for the coefficient, input-state,
and interval derivatives of an affine exponential action. The algorithm
shares one forward Krylov space and one adjoint Krylov space across the
complete operator library.
\item We give a one-workspace quadrature schedule that stores one Krylov basis
and the forward path at the quadrature nodes. A projected-core calculation
can use comparable memory by rebuilding one of its two Krylov bases.
\item We separate the gradient error into forward projection, adjoint
projection, quadrature, and cross terms. Arnoldi residuals bound the two
projection errors, and contractivity gives computable bounds for Markov and
Lindblad generators.
\item We identify one-sided termination and slow convergence as two cases in
which refinement estimates can report a small value before the gradient is
resolved, and we compare these estimates with independent reference
gradients.
\item We use the adjoint Krylov algorithm for sparse parameter recovery in
chemical master, Kuramoto--Sivashinsky, and Lindblad equations.
\end{itemize}

\subsection{Related work}
Arnoldi approximations of matrix-exponential actions are standard tools for
large sparse evolution equations
\cite{moler2003nineteen,higham2008functions,almohy2011action,saad2003iterative,
hochbruck1997krylov,sidje1998expokit}. Fr\'{e}chet derivatives can be evaluated
through a block exponential~\cite{almohy2009frechet}, or approximated by block
and low-rank Krylov methods
\cite{kandolf2017block,kressner2019bivariate,kandolf2021lowrank,
kressner2026krylov}. Kressner~\cite{kressner2019bivariate} introduced the
two-sided Krylov projection used here, and Kandolf et
al.~\cite{kandolf2021lowrank} analyzed the resulting low-rank approximation
and proposed practical stopping estimates. For fixed forward and adjoint
bases, the projected method evaluates the reduced Fr\'{e}chet derivative
directly, whereas our calculation applies quadrature to its path
representation. We use the projected form as a reference calculation and
study the stopping estimates in
Section~\ref{sec:stopping_failure}.

Continuous-adjoint methods compute the gradient of a scalar objective with
one forward and one backward solve~\cite{cao2003adjoint,serban2005cvodes}.
For the autonomous problem in \eqref{eq:intro_evolution}, both solutions can
be represented in Krylov spaces and evaluated at the quadrature nodes from
small projected matrices. Matfree instead differentiates the finite Arnoldi
recurrence~\cite{kraemer2024matfree}. We compare the two approaches at matched
gradient accuracy in Section~\ref{sec:matfree}.

Sparse equation-discovery methods fit time derivatives, increments, or weak
forms to a candidate library
\cite{brunton2016discovering,hoffmann2019reactive,rudy2017data,
schaeffer2017integral,messenger2021weakode,messenger2021weak,
blumensath2009iterative}. We instead fit the finite-time evolution generated
by the library. In the Brusselator comparison of
Section~\ref{sec:weak_comparator}, the observations and candidate terms are
held fixed so that only the fitting method changes.

Data-driven operator methods infer an unrestricted discrete-time or
infinitesimal evolution operator
\cite{williams2015data,schmid2022dynamic,klus2020datadriven,
nejati2021generator,kostic2024generator}, whereas our coefficients multiply a
given mechanistic library. Other inference methods for stochastic reaction
networks use path likelihoods, approximate Bayesian computation, or moment
closures
\cite{golightly2005bayesian,boys2008bayesian,toni2009abc,
neuert2013systematic,singh2011approximate,schnoerr2017approximation}. For
Lindblad dynamics, we use libraries of Hamiltonian commutators and dissipative
channels~\cite{gorini1976completely,lindblad1976generators}, while related
identification methods consider steady-state and restricted observations
\cite{bairey2020steady,gebhart2023learning}.

\subsection{Outline}
Section~\ref{sec:setup} states the affine inverse problem and finite-time
loss. Section~\ref{sec:gradient} derives the gradient calculation and its
cost. Section~\ref{sec:gradient_error} separates the quadrature and projection
errors. Section~\ref{sec:accuracy} studies numerical accuracy, stopping, and
computational cost. Section~\ref{sec:stls} presents the three inverse
problems. The supplementary material contains implementation details and
additional experiments.

\section{Parameter Inference for Evolution Operators}
\label{sec:setup}
We begin with a generator assembled from known library directions and then
define the finite-time loss whose gradient is required by the inference
procedure.

\subsection{Operator library}
\label{sec:affine}

Let $\bm u(t)\in\mathbb C^n$ satisfy
\begin{equation}
\label{eq:affine}
\frac{d\bm u}{dt}=\bA(\btheta)\bm u,
\qquad
\bA(\btheta)=\bA_0+
\sum_{j=1}^{M_{\rm lib}}\theta_j\bA_j.
\end{equation}
The fixed matrix $\bA_0$ and the $M_{\rm lib}$ library directions $\bA_j$
are known, while the real coefficients $\theta_j$ are to be estimated. The
state and generator may be real or complex, while $\btheta$ is subject only
to the sign or feasibility constraints imposed by the model. Differentiating
$\bA(\btheta)$ does not
require sparse coefficients, since sparsity enters only through the outer
optimization used
to select active directions from an overcomplete library.

We index observation intervals by $k$ and quadrature nodes by $q$. The
forward and adjoint Krylov depths are $m_f$ and $m_a$, and we use $m_K$ when
the depths are equal. The symbol $*$ denotes conjugate transpose and reduces
to the ordinary transpose in the real case.

\subsection{Finite-time loss}
Consider an observation interval of length $h$ with initial state $\bv$. Its
final state is
\begin{equation}
\label{eq:layer}
\bm y=e^{h\bA(\btheta)}\bv.
\end{equation}
For a real-valued loss of a complex state, we define gradients using the real
part of the Hermitian product. If $C$ is a linear observation map and $\bm z$
is the observed value, the least-squares loss
\begin{equation*}
 \ell(\bm y)=\frac12\|C\bm y-\bm z\|_2^2
 \qquad\text{has}\qquad
 \nabla_{\bm y}\ell=C^*(C\bm y-\bm z).
\end{equation*}
We start the adjoint calculation from $\nabla_{\bm y}\ell$. For another
differentiable loss, we compute its gradient with respect to $\bm y$ and use
that vector instead.

For full-state probability distributions $\hat\bp_k$ observed at equally
spaced times, we use the one-step prediction loss
\begin{equation}
\label{eq:loss}
\calL(\btheta)=\frac12\sum_{k=0}^{T-1}
\|e^{\bA(\btheta)\Delta t}\hat\bp_k-\hat\bp_{k+1}\|_2^2.
\end{equation}
We can use this calculation for nonuniform observation times, partial
observations, and other differentiable losses. For each observation interval,
we compute its length and the gradient of the loss with respect to the
propagated state, using the appropriate observation map. Each observation
interval has its own forward and adjoint bases, which are shared by all
library directions on that interval. If the
coefficients are outputs of another model, the chain rule maps the coefficient
gradient back to the parameters of that model. The neural context-to-rate
example in the supplementary material illustrates this composition.

\section{Adjoint Krylov Gradient Computation}
\label{sec:gradient}
The coefficient gradient is obtained by pairing the parameter variation of
the forward equation with an adjoint solution. We first derive this identity
and then evaluate it using separate forward and adjoint Krylov spaces.

\subsection{Forward--adjoint formula}
The exponential action in~\eqref{eq:layer} is the endpoint of the linear
initial-value problem
\begin{equation}
\label{eq:forward_path_ode}
 \bm f'(s)=h\bA(\btheta)\bm f(s),\qquad
 \bm f(0)=\bv,\qquad \bm y=\bm f(1).
\end{equation}
When we differentiate this ODE with respect to a coefficient, we obtain a
forced linear ODE. Its solution follows from variation of constants, which
propagates the forcing applied at time $s$ over the remaining interval $1-s$
and integrates these contributions over $0\leq s\leq1$. Pairing this solution
with the gradient of the loss introduces the forward and adjoint paths
\[
 \bm f(s)=e^{sh\bA}\bv,
 \qquad
 \bm\lambda(s)=e^{(1-s)h\bA^*}\nabla_{\bm y}\ell,
\]
and gives the standard continuous-adjoint sensitivity formula
\cite{cao2003adjoint}.

\begin{proposition}[Adjoint gradient formula]
\label{prop:adjoint_gradient}
\label{thm:shared_grad}
For the action~\eqref{eq:layer}, the derivative of the loss with respect to
coefficient $\theta_j$ is
\begin{equation}
\label{eq:general_vjp}
\frac{\partial\ell}{\partial\theta_j}=h\operatorname{Re}\int_0^1
\bm\lambda(s)^*\bA_j\bm f(s)\,ds.
\end{equation}
\end{proposition}
\begin{proof}
Since $\partial\bA/\partial\theta_j=\bA_j$, differentiating
Eq.~\eqref{eq:forward_path_ode} gives
\begin{equation}
\label{eq:tangent_ode}
 \frac{d}{ds}\frac{\partial\bm f(s)}{\partial\theta_j}
 =h\bA\frac{\partial\bm f(s)}{\partial\theta_j}
 +h\bA_j\bm f(s),
 \qquad \frac{\partial\bm f(0)}{\partial\theta_j}=0.
\end{equation}
Applying variation of constants at $s=1$ gives
\begin{equation}
\label{eq:action_directional_derivative}
 \frac{\partial\bm y}{\partial\theta_j}
 =h\int_0^1e^{(1-s)h\bA}\bA_j\bm f(s)\,ds.
\end{equation}
Applying the chain rule to $\ell(\bm y)$ and taking the real part in the
complex case gives
\begin{align}
 \frac{\partial\ell}{\partial\theta_j}
 &=\operatorname{Re}\left(
 (\nabla_{\bm y}\ell)^{*}\frac{\partial\bm y}{\partial\theta_j}
 \right) \notag\\
 &=h\operatorname{Re}\int_0^1
 (\nabla_{\bm y}\ell)^{*}e^{(1-s)h\bA}\bA_j
 e^{sh\bA}\bv\,ds \notag\\
 &=h\operatorname{Re}\int_0^1
 \left(e^{(1-s)h\bA^*}\nabla_{\bm y}\ell\right)^*
 \bA_j\left(e^{sh\bA}\bv\right)\,ds \notag\\
 &=h\operatorname{Re}\int_0^1
 \bm\lambda(s)^*\bA_j\bm f(s)\,ds.
\label{eq:shared_gradient_derivation}
\end{align}
The last equality uses
$(\nabla_{\bm y}\ell)^*e^{(1-s)h\bA}
=(e^{(1-s)h\bA^*}\nabla_{\bm y}\ell)^*$.
\end{proof}

The formula separates propagation from the library direction. We compute
$\bm f$ and $\bm\lambda$ once, while $\bA_j$ enters only through the bilinear
product $\bm\lambda(s)^*\bA_j\bm f(s)$.

The coefficient derivative can also be viewed as a Fr\'{e}chet derivative,
which measures the first-order change of a matrix function in a direction
$\bE$. Formally, $L_f(\bA,\bE)$ is defined by
$f(\bA+\varepsilon\bE)=f(\bA)+\varepsilon L_f(\bA,\bE)+o(\varepsilon)$.
For the exponential, it satisfies~\cite[Eq.~(10.15)]{higham2008functions}
\begin{equation}
\label{eq:frechet_exp}
 L_{\exp}(h\bA,h\bE)
 =h\int_0^1e^{(1-s)h\bA}\bE e^{sh\bA}\,ds.
\end{equation}
Evaluating this integral separately for every $\bA_j$ repeats the two
exponential propagations. Pairing the derivative with the loss gradient
replaces these repeated evaluations by the shared forward and adjoint paths.

For interval $k$ of the snapshot loss~\eqref{eq:loss}, the prediction error is
\[
 \br_k=e^{\bA\Delta t}\hat\bp_k-\hat\bp_{k+1}.
\]
The forward path begins at $\hat\bp_k$, while the adjoint path begins at
$\br_k$ at the right endpoint of the interval. Thus,
\[
 \mathbf f_k(s)=e^{s\bA\Delta t}\hat\bp_k,
 \qquad
 \blambda_k(s)=e^{(1-s)\bA^T\Delta t}\br_k.
\]
For one transition, $\ell_k=\tfrac12\|\br_k\|_2^2$, with differential
\begin{align}
 d\ell_k
 &=\br_k^T d\br_k
 =\br_k^T d\left(e^{\bA\Delta t}\hat\bp_k\right).
\label{eq:snapshot_loss_differential}
\end{align}
The differential shows that the adjoint path for interval $k$ starts from
$\br_k$ at its right endpoint. Substitution into
Proposition~\ref{thm:shared_grad}, followed by summation over the observation
intervals, gives
\begin{equation}
\label{eq:shared_grad}
\frac{\partial\calL}{\partial\theta_j}
=\Delta t\sum_{k=0}^{T-1}\int_0^1
\blambda_k(s)^T\bA_j\mathbf f_k(s)\,ds.
\end{equation}
For another loss, $\br_k$ is replaced by the gradient of that interval's loss
with respect to its predicted state. Figure~\ref{fig:shared_arnoldi}
summarizes the two stages, and Algorithm~\ref{alg:shared_arnoldi} gives the
calculation for one exponential action. The forward basis produces the
prediction, whose loss gradient starts the adjoint Arnoldi process. If the
input state or this loss gradient is zero, we take the corresponding path to
be zero and do not normalize its starting vector.

\subsection{Krylov approximation and quadrature}

The forward and adjoint paths require exponential actions of $\bA$ and
$\bA^*$. Arnoldi projection replaces each large exponential by the
exponential of a small upper-Hessenberg matrix. For an operator $B$ and a
starting vector $v$, the associated Krylov space is
\[
 \mathcal K_m(B,v)=\operatorname{span}\{v,Bv,\ldots,B^{m-1}v\}
\].
Suppose that $v=\beta v_1$, where $\beta=\|v\|_2$. After $m$ Arnoldi
steps, the columns of $V_m=[v_1,\ldots,v_m]$ form an orthonormal basis for
this space, and the projected matrix $H_m=V_m^*BV_m$ is upper Hessenberg
\cite{hochbruck1997krylov,sidje1998expokit}. If $\bdelta_k$ denotes the
$k$th coordinate vector, the Arnoldi relation and the resulting exponential
approximation are
\begin{equation}
\label{eq:arnoldi_recurrence}
B V_m=V_m H_m+h_{m+1,m}v_{m+1}\bdelta_m^*,
\qquad
e^{\tau B}v\approx \beta V_m e^{\tau H_m}\bdelta_1.
\end{equation}
The second relation follows by approximating the solution of
$u'(\tau)=Bu(\tau)$, $u(0)=v$, in the form
$u_m(\tau)=V_mz_m(\tau)$. Requiring the ODE residual to be orthogonal to
$\mathcal K_m(B,v)$ gives
\begin{align}
 V_m^*\bigl(V_m z_m'(\tau)-BV_mz_m(\tau)\bigr)&=0,\notag\\
 z_m'(\tau)&=H_mz_m(\tau),
 \qquad z_m(0)=\beta\bdelta_1.
\label{eq:projected_evolution}
\end{align}
Thus $z_m(\tau)=\beta e^{\tau H_m}\bdelta_1$, and the large exponential
action is computed from an $m\times m$ exponential followed by a
multiplication by $V_m$. The Arnoldi relation also shows that the defect in
the original ODE lies in the direction $v_{m+1}$. Its norm is
\begin{equation}
\label{eq:arnoldi_residual}
\rho_m(\tau)=\beta |h_{m+1,m}|\,
|\bdelta_m^*e^{\tau H_m}\bdelta_1|.
\end{equation}
For a strongly nonnormal matrix, a small residual does not by itself imply a
small error because the remaining evolution may amplify the defect.
Section~\ref{sec:gradient_error} accounts for this amplification in the error
bounds. We also compare successive Krylov depths as a numerical diagnostic.

We construct the forward and adjoint spaces independently, with depths $m_f$
and $m_a$. When the depths are equal, we denote their common value by $m_K$.
For snapshot interval $k$, let
\[
(\bV_k^f,\bH_k^f,\beta_k^f)
\leftarrow
\mathrm{arnoldi}(\bA,\hat{\bp}_k;m_f),
\qquad
(\bV_k^a,\bH_k^a,\beta_k^a)
\leftarrow
\mathrm{arnoldi}(\bA^T,\br_k;m_a),
\]
where $\beta_k^f=\|\hat{\bp}_k\|_2$ and $\beta_k^a=\|\br_k\|_2$. Since the first
Arnoldi vector is the normalized starting vector,
$\bV_k^f{}^T\hat{\bp}_k = \beta_k^f\bdelta_1$ and $\bV_k^a{}^T\br_k = \beta_k^a\bdelta_1$.
The projected forward and adjoint paths are therefore
\[
\mathbf{f}_k(s)\approx \beta_k^f \bV_k^f e^{s\Delta t\,\bH_k^f}\bdelta_1,
\qquad
\blambda_k(s)\approx \beta_k^a \bV_k^a e^{(1-s)\Delta t\,\bH_k^a}\bdelta_1.
\]
At the $q$th Gauss--Legendre node, the projected paths are
\[
\mathbf{f}_k^{(q)}=\beta_k^f\bV_k^f e^{s_q\Delta t\,\bH_k^f}\bdelta_1,
\qquad
\blambda_k^{(q)}=\beta_k^a\bV_k^a e^{(1-s_q)\Delta t\,\bH_k^a}\bdelta_1.
\]
Using the corresponding weights $w_q$ in~\eqref{eq:shared_grad} gives
\begin{equation}
\label{eq:quadrature}
\frac{\partial\calL}{\partial\theta_j}
\approx
\Delta t
\sum_{k=0}^{T-1}\sum_{q=1}^{n_q}
w_q\operatorname{Re}\!\left[
(\blambda_k^{(q)})^* \bA_j \mathbf{f}_k^{(q)}
\right].
\end{equation}
For fixed $k$ and $q$, the two path vectors are independent of the library
index $j$. We compute them once and then evaluate the contractions
\begin{equation}
\label{eq:direction_contraction}
 c_{k,q}^{(j)}
 = (\blambda_k^{(q)})^*\bA_j\mathbf f_k^{(q)}
\end{equation}
for all library directions. When $\bA_j$ is sparse, each contraction requires
one sparse matrix--vector product and one inner product. Enlarging the library
therefore adds contractions but does not require new Krylov spaces or
exponential actions. If $\bA_j$ commutes with $\bA$, then
$\blambda(s)^*\bA_j\mathbf f(s)$ is independent of $s$, so a one-node rule
introduces no quadrature error.

Figure~\ref{fig:shared_arnoldi} separates the calculation into two steps. We
first compute the forward and adjoint paths for each observation interval. We
then contract these paths with every library direction.

\begin{figure}[t]
\centering
\resizebox{\linewidth}{!}{%
\begin{tikzpicture}[
  >=stealth, semithick, font=\footnotesize,
  top/.style={draw, rounded corners=3pt, fill=white,
              minimum width=4.9cm, minimum height=1.6cm,
              text width=4.7cm, align=center, inner sep=5pt},
  bot/.style={draw, rounded corners=3pt, fill=white,
              minimum width=11.0cm, minimum height=1.0cm,
              text width=10.8cm, align=center, inner sep=5pt},
  hdr/.style={font=\footnotesize\bfseries},
  ann/.style={font=\scriptsize, text=black!60, align=center}
]
\node[hdr, text width=12cm, align=center] (h1) at (5.5,2.1)
  {Step 1: compute the forward and adjoint paths\quad
   {\normalfont\itshape once per observation interval}};
\node[top] (fwd) at (2.4,0.55) {
  \textbf{Forward path}\\[3pt]
  $(\bV_k^f,\,\bH_k^f)
     \leftarrow \mathrm{arnoldi}(\bA,\,\hat\bp_k;\,m_f)$\\[2pt]
  $\mathbf{f}_k^{(q)}
     =\beta_k^f\bV_k^f e^{s_q\Delta t\,\bH_k^f}\bdelta_1$
};
\node[top] (adj) at (8.6,0.55) {
  \textbf{Adjoint path}\\[3pt]
  $(\bV_k^a,\,\bH_k^a)
     \leftarrow \mathrm{arnoldi}(\bA^*,\,\br_k;\,m_a)$\\[2pt]
  $\blambda_k^{(q)}
     =\beta_k^a\bV_k^a e^{(1-s_q)\Delta t\,\bH_k^a}\bdelta_1$
};
\node[ann] at (5.5,-0.55)
  {\textit{the two paths depend on interval $k$, not on the library direction $j$}};
\node[bot] (dir) at (5.5,-2.0) {
  $\dfrac{\partial\calL}{\partial\theta_j}
   \approx\Delta t\sum_{k=0}^{T-1}\sum_{q=1}^{n_q}
   w_q(\blambda_k^{(q)})^T\bA_j\mathbf f_k^{(q)}$
  \enskip for all $j$
};
\node[hdr, below=0.15cm of dir, text width=12cm, align=center]
  {Step 2: contract the paths with the library directions\quad
   {\normalfont\itshape all $M_{\rm lib}$ directions,\quad $O(\nnz(\bA_j)+n)$ per direction per $(k,q)$}};
\draw[->] (fwd.south) to[out=270,in=130] ([xshift=-0.5cm]dir.north);
\draw[->] (adj.south) to[out=270,in=50] ([xshift=0.5cm]dir.north);
\end{tikzpicture}}
\caption{The forward and adjoint paths are computed once for each observation
interval and then reused for every library direction.}
\label{fig:shared_arnoldi}
\end{figure}

\begin{algorithm}[t]
\caption{Adjoint Krylov gradient for one exponential action}
\label{alg:shared_arnoldi}
\begin{algorithmic}[1]
\REQUIRE Generator $\bA$, directions $\{\bA_j\}$, state $\bv$, interval $h$, differentiable loss $\ell(\bm y)$, depths $(m_f,m_a)$, Gauss--Legendre nodes and weights $\{s_q,w_q\}$ on $[0,1]$
\ENSURE Prediction $\bm y$ and derivatives $\nabla_{\btheta}\ell$, $\nabla_{\bv}\ell$, and $\partial\ell/\partial h$
\STATE $(V_f,H_f,\beta_f)\leftarrow\mathrm{arnoldi}(\bA,\bv;m_f)$
\STATE $\bm y\leftarrow\beta_f V_f e^{hH_f}\bdelta_1$ and compute the loss gradient $\bm g_y\leftarrow\nabla_{\bm y}\ell(\bm y)$
\STATE $(V_a,H_a,\beta_a)\leftarrow\mathrm{arnoldi}(\bA^*,\bm g_y;m_a)$
\STATE $\nabla_{\bv}\ell\leftarrow\beta_a V_a e^{hH_a}\bdelta_1$
\FOR{$q=1,\ldots,n_q$}
\STATE $\bm f_q\leftarrow\beta_f V_f e^{s_q hH_f}\bdelta_1$,
$\bm\lambda_q\leftarrow\beta_a V_a e^{(1-s_q)hH_a}\bdelta_1$
\ENDFOR
\FOR{$j=1,\ldots,M_{\rm lib}$}
\STATE $\partial\ell/\partial\theta_j\leftarrow h\operatorname{Re}\sum_qw_q\bm\lambda_q^*\bA_j\bm f_q$
\ENDFOR
\STATE $\partial\ell/\partial h\leftarrow\operatorname{Re}(\bm g_y^*\bA\bm y)$
\end{algorithmic}
\end{algorithm}

The coefficient gradient is sufficient when the input state and interval
length are fixed. When the exponential action is part of a larger model,
either quantity may also depend on unknown parameters. We then need
derivatives with respect to $\bm v$ and $h$. For
$\bm y=e^{h\bA(\btheta)}\bm v$, a perturbation of the input state gives
\begin{align}
 d\bm y&=e^{h\bA}d\bm v,\notag\\
 d\ell
 &=\operatorname{Re}\bigl((\nabla_{\bm y}\ell)^*e^{h\bA}d\bm v\bigr)
 =\operatorname{Re}\bigl((e^{h\bA^*}\nabla_{\bm y}\ell)^*d\bm v\bigr).
\end{align}
It follows that
\begin{equation}
\label{eq:composable_pullback}
\nabla_{\bm v}\ell=e^{h\bA^*}\nabla_{\bm y}\ell.
\end{equation}
Differentiating with respect to the interval length gives
\begin{align}
 \frac{\partial\bm y}{\partial h}
 &=\bA e^{h\bA}\bm v=\bA\bm y,\notag\\
 \frac{\partial\ell}{\partial h}
 &=\operatorname{Re}\left((\nabla_{\bm y}\ell)^*
 \frac{\partial\bm y}{\partial h}\right)
 =\operatorname{Re}\bigl((\nabla_{\bm y}\ell)^*\bA\bm y\bigr).
\label{eq:interval_pullback}
\end{align}
The adjoint path at $s=0$ gives the input-state derivative in
\eqref{eq:composable_pullback}. The interval derivative requires one
application of $\bA$ to the computed output $\bm y$. Thus, both derivatives
use quantities already computed by Algorithm~\ref{alg:shared_arnoldi}.

The input-state derivative is needed when an exponential action follows a
nonlinear calculation. For example, consider the semilinear equation
\[
 \dot{\bm u}=\bA(\btheta)\bm u+\bm N(\bm u),
\]
and the first-order integrating-factor step
\begin{equation}
\label{eq:semilinear_step}
\bm u_{k+1}=e^{h\bA(\btheta)}
\bigl(\bm u_k+h\bm N(\bm u_k)\bigr).
\end{equation}
Here the input to the exponential is
$\bm v_k=\bm u_k+h\bm N(\bm u_k)$. Equation
\eqref{eq:composable_pullback} gives the derivative of the loss with respect
to $\bm v_k$. We then apply the chain rule to obtain derivatives with respect
to $\bm u_k$ and any unknown coefficients in $\bm N$. The first-order formula
also introduces a time-discretization error in addition to the Krylov and
quadrature errors. Section~\ref{sec:ks_exp} measures the effect of the
input-state derivative in the Kuramoto--Sivashinsky experiment.

\subsection{Cost and storage}
The work has three parts. We construct the forward and adjoint Krylov spaces,
compute the two paths at the quadrature nodes, and contract the paths with the
library directions. Suppose that Algorithm~\ref{alg:shared_arnoldi} is applied
$N_{\rm act}$ times with a common Krylov depth $m_K$ and $n_q$ quadrature
nodes.
\begin{proposition}[Arithmetic cost]
\label{prop:complexity}
Assume full Arnoldi orthogonalization, dense arithmetic for the projected
exponentials, and direct sparse computation of every library contraction.
Then the arithmetic cost is
\begin{align*}
&O\bigl(N_{\rm act}(m_K\nnz(\bA)+nm_K^2)\bigr)\\
&\quad+O\bigl(N_{\rm act}n_q(m_K^3+nm_K)\bigr)
+O\!\left(N_{\rm act}n_q\sum_{j=1}^{M_{\rm lib}}(\nnz(\bA_j)+n)\right).
\end{align*}
\end{proposition}
\begin{proof}
Each action constructs two $m_K$-step Arnoldi spaces. The sparse products
with $\bA$ cost $O(m_K\nnz(\bA))$, while full orthogonalization costs
$O(nm_K^2)$. At each quadrature node, computing the two small exponentials
and multiplying by the Arnoldi bases costs $O(m_K^3+nm_K)$. The contraction
with direction $j$ costs $O(\nnz(\bA_j)+n)$. Applying these costs to all
actions, nodes, and library directions gives the stated expression.
\end{proof}
The first term is the cost of the two Arnoldi spaces, the second computes the
paths at the quadrature nodes, and the third contracts those paths with the
library. Once $\bA$ has been assembled, only the third term depends explicitly
on the number of library directions. A larger library may also change the
sparsity of $\bA$. In that case, the first term grows through
$\nnz(\bA)$, and assembling $\bA$ has an additional library-dependent cost.
For CME directions, $\nnz(\bA_j)=O(n)$, so the contraction cost becomes
$O(N_{\rm act}n_qM_{\rm lib}n)$. The supplementary material gives the full
derivation.

Full orthogonalization can be noticeable when multiplication by $\bA$ is
inexpensive, even when $m_K\ll n$. The second term also assumes that a dense
projected exponential is computed separately at every node. A Schur
decomposition of each projected matrix replaces these computations by one
$O(m_K^3)$ factorization and lower-order work at each node.

A direct implementation retains both Arnoldi bases and uses
$O(n(m_f+m_a)+m_f^2+m_a^2)$ storage. Quadrature admits a smaller schedule.
Once the forward basis has been constructed, we evaluate and retain the
forward path at the $n_q$ quadrature nodes. We then overwrite the forward
basis with the adjoint basis and evaluate the adjoint path one node at a
time. The resulting storage is
$O(n(\max(m_f,m_a)+n_q)+m_f^2+m_a^2)$. In contrast, the projected-core
contraction couples the two bases and must retain both when each basis is
constructed only once. It can reach the one-workspace memory scale by
rebuilding one basis after forming the reduced core. The supplement compares
these three schedules on a matrix-free Lindblad problem. These estimates do
not include the operator library.

\section{Gradient Error Analysis}
\label{sec:gradient_error}

The gradient approximation has three sources of error. Two come from the
forward and adjoint Krylov spaces, and the third comes from quadrature. We
first separate these errors and then bound the two path errors using the
Arnoldi residuals. A small residual need not imply a small error when $\bA$
is nonnormal because the remaining evolution may amplify the defect. The
bounds include this amplification and the size of each library direction.
They are intended to detect an unresolved path, not to select the smallest
Krylov depth on every stiff problem.

Let
\[
 \bm f(s)=e^{sh\bA}\bv,
 \qquad
 \bm\lambda(s)=e^{(1-s)h\bA^*}\nabla_{\bm y}\ell,
 \qquad 0\leq s\leq1,
\]
and let $\widetilde{\bm f}$ and $\widetilde{\bm\lambda}$ denote their Krylov
approximations. The approximate paths have the exact boundary values at
$s=0$ and $s=1$, respectively. For any path $\bm z$, we write
$\|\bm z\|_\infty=\max_{0\leq s\leq1}\|\bm z(s)\|_2$. The quadrature nodes
and nonnegative weights are denoted by $s_q,w_q$, with $\sum_qw_q=1$.

Before deriving the bound, we record why the two path errors must be treated
separately. Consider a rank-one direction
$\bE=\bm a\bm b^*$ with $\bm a=\alpha\bm v_1$ and
$\bm b=\beta\bm w_1$. Let the two Arnoldi relations be
\begin{align}
 \bA V_m&=V_mP_m+p_{m+1,m}\bm v_{m+1}\bdelta_m^*,\notag\\
 \bA^*W_m&=W_mQ_m+q_{m+1,m}\bm w_{m+1}\bdelta_m^*.
\label{eq:two_sided_arnoldi}
\end{align}
Here $p_{m+1,m}$ and $q_{m+1,m}$ are denoted $g_{m+1,m}$ and
$h_{m+1,m}$, respectively, in \cite{kandolf2021lowrank}.
The depth-$m$ projected Fr\'{e}chet approximation of
\cite[Algorithm~2]{kandolf2021lowrank} is
\begin{equation}
 L_m=V_mX_mW_m^*,\qquad
 X_m=\alpha\beta\int_0^1 e^{(1-s)P_m}
 \bdelta_1\bdelta_1^*e^{sQ_m^*}\,ds.
\label{eq:two_sided_frechet}
\end{equation}
The residual-product estimate in~\cite[Eq.~(59)]{kandolf2021lowrank} has the form
\begin{equation}
 \eta_m^{\times}
 =|p_{m+1,m}q_{m+1,m}\,\xi_m|,
 \label{eq:multiplicative_estimator}
\end{equation}
where $\xi_m$ is an entry of a reduced block matrix function. This estimate is
used as a heuristic in that work.

\begin{proposition}[One-sided termination]
\label{prop:one_sided_termination}
Suppose that $\bA^*\bm b=\overline{\lambda}\bm b$ and let $V_m$ be the
Arnoldi basis constructed from $(\bA,\bm a)$, with
$P_m=V_m^*\bA V_m$. Define
\[
 \psi_\lambda(z)=
 \begin{cases}
  (e^z-e^\lambda)/(z-\lambda),&z\ne\lambda,\\
  e^\lambda,&z=\lambda.
 \end{cases}
\]
Then the Arnoldi process for $(\bA^*,\bm b)$ terminates after one step and
the residual-product estimate in~\eqref{eq:multiplicative_estimator} is zero.
The exact derivative and the depth-$(m,1)$ projected approximation satisfy
\begin{align}
 L_{\exp}(\bA,\bm a\bm b^*)&=\psi_\lambda(\bA)\bm a\bm b^*,\notag\\
 L_{m,1}&=V_m\psi_\lambda(P_m)V_m^*\bm a\bm b^*.
\label{eq:one_sided_reduction}
\end{align}
Consequently,
\begin{equation}
 \left\|L_{\exp}(\bA,\bm a\bm b^*)-L_{m,1}\right\|_F
 =\|\bm b\|_2
 \left\|\psi_\lambda(\bA)\bm a
 -V_m\psi_\lambda(P_m)V_m^*\bm a\right\|_2.
\label{eq:one_sided_error}
\end{equation}
Thus one-sided termination makes the residual-product estimate zero while
leaving the ordinary Krylov approximation error on the other side unchanged.
The same conclusion holds with the two Arnoldi processes interchanged.
\end{proposition}

\begin{proof}
The eigenvector relation gives
$\bm b^*e^{s\bA}=e^{s\lambda}\bm b^*$. Substitution into the integral
definition of the Fr\'{e}chet derivative gives the first identity in
\eqref{eq:one_sided_reduction}. Applying the same calculation to the projected
matrix $P_m$ gives the second identity. Equation~\eqref{eq:one_sided_error}
then follows from the Frobenius norm of a rank-one matrix. Since the Arnoldi
residual coefficient on the $\bm b$ side is zero, so is
$\eta_m^{\times}$.
\end{proof}

\begin{theorem}[Gradient error decomposition]
\label{thm:residual_gradient_error}
For each library direction, let
\[
 \widetilde g_j
 =h\operatorname{Re}\sum_qw_q
 \widetilde{\bm\lambda}(s_q)^*\bA_j\widetilde{\bm f}(s_q).
\]
Then
\begin{align}
 \left|\frac{\partial\ell}{\partial\theta_j}-\widetilde g_j\right|
 &\leq |h|\left|\operatorname{Re}\left(
 \int_0^1\bm\lambda(s)^*\bA_j\bm f(s)\,ds
 -\sum_qw_q\bm\lambda(s_q)^*\bA_j\bm f(s_q)
 \right)\right|\notag\\
 &\quad+|h|\|\bA_j\|_2\Bigl(
 \|\bm f\|_\infty\|\bm\lambda-\widetilde{\bm\lambda}\|_\infty
 +\|\bm\lambda\|_\infty\|\bm f-\widetilde{\bm f}\|_\infty\notag\\
 &\hspace{43mm}
 +\|\bm\lambda-\widetilde{\bm\lambda}\|_\infty
  \|\bm f-\widetilde{\bm f}\|_\infty\Bigr).
 \label{eq:main_path_gradient_bound}
\end{align}
\end{theorem}

The first term is the quadrature error for the exact paths. The remaining
terms contain the forward path error, the adjoint path error, and their
product. Each is scaled by the norm of the library direction.

\begin{proof}
Let $e_f=\bm f-\widetilde{\bm f}$ and
$e_\lambda=\bm\lambda-\widetilde{\bm\lambda}$. After adding and subtracting
the quadrature sum for the exact paths, the difference at each node is
\[
 \bm\lambda^*\bA_j\bm f-
 \widetilde{\bm\lambda}^*\bA_j\widetilde{\bm f}
 =e_\lambda^*\bA_j\bm f+\bm\lambda^*\bA_je_f
 -e_\lambda^*\bA_je_f.
\]
Applying $|\bm x^*\bA_j\bm y|\leq
\|\bm x\|_2\|\bA_j\|_2\|\bm y\|_2$ to these three terms gives a bound that is
independent of the node. Summing with nonnegative weights of unit total mass
proves~\eqref{eq:main_path_gradient_bound}.
\end{proof}

To bound the two path errors in
Theorem~\ref{thm:residual_gradient_error}, define the Arnoldi residuals
\begin{equation}
 \bm r_f(s)=\widetilde{\bm f}'(s)-h\bA\widetilde{\bm f}(s),
 \qquad
 \bm r_\lambda(s)=\widetilde{\bm\lambda}'(s)
 +h\bA^*\widetilde{\bm\lambda}(s).
 \label{eq:main_path_residuals}
\end{equation}
Subtracting the exact and approximate evolution equations and applying
variation of constants gives
\begin{align*}
 \bm f(s)-\widetilde{\bm f}(s)
 &=-\int_0^s e^{(s-u)h\bA}\bm r_f(u)\,du,\\
 \bm\lambda(s)-\widetilde{\bm\lambda}(s)
 &=\int_s^1 e^{(u-s)h\bA^*}\bm r_\lambda(u)\,du.
\end{align*}
The finite-interval amplification factor is
\[
 \kappa_A=\max_{0\leq s\leq1}\int_0^s
 \|e^{(s-u)h\bA}\|_2\,du.
\]
The adjoint has the same factor in the $2$-norm after a change of variables.
Consequently,
\[
 \|\bm f-\widetilde{\bm f}\|_\infty
 \leq\kappa_A\|\bm r_f\|_\infty,
 \qquad
 \|\bm\lambda-\widetilde{\bm\lambda}\|_\infty
 \leq\kappa_A\|\bm r_\lambda\|_\infty.
\]
The exact path norms in Theorem~\ref{thm:residual_gradient_error} satisfy
\[
 \|\bm f\|_\infty\leq\|\widetilde{\bm f}\|_\infty
 +\kappa_A\|\bm r_f\|_\infty,
 \qquad
 \|\bm\lambda\|_\infty\leq\|\widetilde{\bm\lambda}\|_\infty
 +\kappa_A\|\bm r_\lambda\|_\infty.
\]
Substitution into~\eqref{eq:main_path_gradient_bound} bounds the two path
errors using the forward and adjoint residuals, $\kappa_A$, the computed path
norms, and $\|\bA_j\|_2$. The quadrature error remains separate.

For the $m_f$-step Arnoldi path
$\widetilde{\bm f}(s)=\beta_fV_{m_f}e^{shH^f_{m_f}}\bdelta_1$, we use the Arnoldi
relation to obtain the computable residual
\begin{equation}
 \|\bm r_f(s)\|_2
 =|h|\beta_f|h^f_{m_f+1,m_f}|
 |\bdelta_{m_f}^*e^{shH^f_{m_f}}\bdelta_1|,
 \label{eq:main_arnoldi_residual}
\end{equation}
The adjoint residual has the analogous form. This is the ODE residual used in
residual-based error control for Krylov exponential methods
\cite{botchev2013residual}. The two Arnoldi factorizations therefore compute
the residual quantities needed in~\eqref{eq:main_path_gradient_bound}. The
remaining factor measures the growth of $e^{t\bA}$ over the interval.

\begin{corollary}[Contractive CTMC bound]
\label{cor:ctmc_residual_bound}
Consider a CTMC generator or sub-generator $\bA$ in the column convention,
with nonnegative off-diagonal entries and nonpositive column sums, and let
$h\geq0$.
Let
\[
 \delta_f(s)=\int_0^s\|\bm r_f(u)\|_1\,du,\qquad
 \delta_a(s)=\int_s^1\|\bm r_\lambda(u)\|_\infty\,du.
\]
For each library direction, let
\[
 g_j^{\mathrm{proj}}=h\operatorname{Re}\int_0^1
 \widetilde{\bm\lambda}(s)^*\bA_j\widetilde{\bm f}(s)\,ds.
\]
Then
\begin{align}
 \left|\frac{\partial\ell}{\partial\theta_j}
 -g_j^{\mathrm{proj}}\right|
 &\leq h\int_0^1\Bigl(
 \delta_a(s)\|\bA_j\widetilde{\bm f}(s)\|_1
 +\delta_f(s)\|\bA_j^T\widetilde{\bm\lambda}(s)\|_\infty\notag\\
 &\hspace{35mm}
 +\|\bA_j\|_1\delta_f(s)\delta_a(s)\Bigr)\,ds.
 \label{eq:ctmc_residual_bound}
\end{align}
\end{corollary}

\begin{proof}
A CTMC or sub-Markov semigroup is contractive in the $1$-norm, and its
adjoint is contractive in the $\infty$-norm. Variation of constants therefore
gives
$\|\bm f(s)-\widetilde{\bm f}(s)\|_1\leq\delta_f(s)$ and
$\|\bm\lambda(s)-\widetilde{\bm\lambda}(s)\|_\infty\leq\delta_a(s)$.
Adding and subtracting the bilinear form of the two projected paths leaves
three terms bounded by
\[
 \delta_a\|\bA_j\widetilde{\bm f}\|_1,\qquad
 \delta_f\|\bA_j^T\widetilde{\bm\lambda}\|_\infty,\qquad
 \|\bA_j\|_1\delta_f\delta_a.
\]
Integration proves~\eqref{eq:ctmc_residual_bound}.
\end{proof}

The Arnoldi factorizations also give the residual norms used in
$\delta_f$ and $\delta_a$. Denote the forward projected matrix and residual
coefficient by $H^f_{m_f}$ and $h^f_{m_f+1,m_f}$, and denote the adjoint
quantities by $H^a_{m_a}$ and $h^a_{m_a+1,m_a}$. If the starting-vector norms
are $\beta_f$ and $\beta_a$, then
\begin{align*}
 \|\bm r_f(u)\|_1
 &=h\beta_f|h^f_{m_f+1,m_f}|\,\|\bv_{m_f+1}\|_1
 \bigl|\bdelta_{m_f}^*e^{uhH^f_{m_f}}\bdelta_1\bigr|,\\
 \|\bm r_\lambda(u)\|_\infty
 &=h\beta_a|h^a_{m_a+1,m_a}|\,\|\bw_{m_a+1}\|_\infty
 \bigl|\bdelta_{m_a}^*e^{(1-u)hH^a_{m_a}}\bdelta_1\bigr|.
\end{align*}
If quadrature is used for the projected integral, its remainder must be added
to the right-hand side of~\eqref{eq:ctmc_residual_bound}. In the numerical
experiments, we compare two quadrature orders instead. This comparison checks
quadrature convergence but is not an error bound.

The same argument applies to a Lindblad semigroup when the dissipation rates
are nonnegative at the evaluation point. The semigroup is contractive in
trace norm, while its adjoint is unital and contractive in operator norm. We
reshape each Arnoldi residual vector as a matrix before computing these norms
\cite{russo1966unitary}.

The three error sources may be refined independently. Increasing the
quadrature order reduces the first term in
\eqref{eq:main_path_gradient_bound}. Increasing the forward or adjoint depth
reduces the corresponding residual. The factors $\kappa_A$ and
$\|\bA_j\|_2$ account for amplification and for the scale of the library
direction. Since $\kappa_A$ can be difficult to estimate for a large
nonnormal operator, Section~\ref{sec:accuracy} also compares successive
Krylov depths and quadrature orders.

\section{Accuracy and Computational Benchmarks}
\label{sec:accuracy}

We first determine the forward depth, adjoint depth, and quadrature order
needed for an accurate gradient. We then compare runtimes using only
configurations that meet the same tolerance. The two Krylov spaces and the
quadrature rule are refined separately. Their errors may partly cancel, so
the total error need not decrease whenever one parameter is increased.

\subsection{Reference and error measure}

For problems small enough to admit a dense reference, we use the block
identity
\cite[Thm.~3.6]{higham2008functions}
\begin{equation}
\label{eq:block_frechet}
\exp\!\begin{pmatrix}\bA\Delta t & \bE\Delta t\\ \bm{0} & \bA\Delta t\end{pmatrix}
=\begin{pmatrix} e^{\bA\Delta t} & L_{\exp}(\bA\Delta t,\bE\Delta t)\\ \bm{0} & e^{\bA\Delta t}\end{pmatrix},
\end{equation}
for each direction $\bE$. We compute the block exponential without Krylov
projection or quadrature and report the relative $2$-norm error of the full
parameter gradient. Finite differences do not resolve the smallest errors in
these tests. With a step of $10^{-6}$, their resolution is approximately
$10^{-10}$.

\subsection{Structural failures of heuristic stopping estimates}
\label{sec:stopping_failure}

Two common diagnostics are not error bounds. The residual-product estimate
vanishes if either Arnoldi process terminates, even when the other path is
unresolved. Figure~\ref{fig:one_sided_failure} shows a 200-state
reflecting-walk example in which this estimate is zero while the
true Fr\'{e}chet error ranges from $20.66$ to $0.137$. A difference between
successive projected approximations can also be small during slow convergence.
For the stiff 1681-state Brusselator example, this difference is
$5.40\times10^{-9}$ at depth 16, while the normalized library-gradient error
is $6.40\times10^{-7}$. The residual bound in
Corollary~\ref{cor:ctmc_residual_bound} remains valid in both cases, although
it can be conservative. The supplement gives both constructions and the
complete plot. We use the heuristic quantities only as diagnostics and check
the forward, adjoint, and quadrature errors separately.

\begin{figure}[tbp]
\centering
\includegraphics[width=0.60\linewidth,trim=0 0 252 0,clip]
  {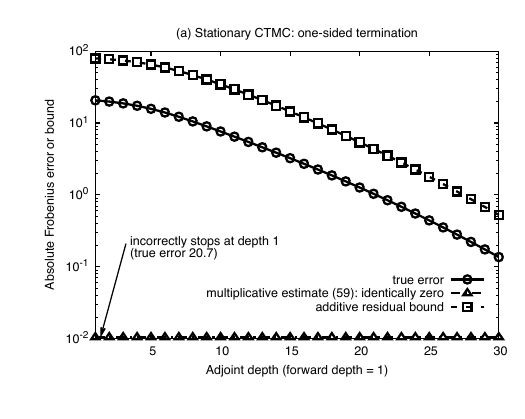}
\caption{One-sided Arnoldi termination makes the residual-product estimate
identically zero although the Fr\'{e}chet error is nonzero. The additive
residual bound remains valid. The supplement gives the stiff-CME test and the
full experimental details.}
\label{fig:one_sided_failure}
\end{figure}

\subsection{Krylov and quadrature errors}

We vary $m_K$ and $n_q$ separately for a stiff nonnormal Fokker--Planck
generator with $n=80$ and $M_{\rm lib}=10$. The supplement gives the full
error table. At $n_q=8$, the error is nearly independent of $m_K$. It changes from
$9.1\times10^{-6}$ at $m_K=20$ to $9.7\times10^{-6}$ at $m_K=60$. The
quadrature remainder therefore dominates. In contrast, at $n_q=24$, increasing
$m_K$ from 20 to 60 reduces the error from $1.0\times10^{-5}$ to
$1.8\times10^{-15}$. The slight increase between $n_q=8$ and $n_q=24$ at
$m_K=20$ reflects the nonmonotonicity of the combined projection and
quadrature errors. The error continues to decrease when both parameters are
refined. Since the integrand is analytic in $s$, Gauss--Legendre quadrature
converges rapidly until the Krylov error or roundoff dominates. Increasing
$n_q$ computes additional small exponential actions but requires no new
Arnoldi steps.

For the complex Lindblad generator, the gradient error reaches $10^{-16}$ at
$n_q=12$ for every tested step size. This experiment also verifies the use of
the conjugate transpose for complex operators.

\subsection{Choosing \texorpdfstring{$m_K$ and $n_q$}{mK and nq} in practice}

Given an action tolerance $\varepsilon_{\rm act}$, we increase $m_f$ until both
$\rho_{m_f}(h)\leq\varepsilon_{\rm act}\|\bv\|_2$ and the forward action agrees
with the next tested depth. Holding that forward action fixed, we increase
$m_a$ until the corresponding checks pass for the loss gradient at the right
endpoint. We then double $n_q$ until the coefficient gradient changes by no
more than the target gradient tolerance. The comparison uses the largest
relative change over all reported directions, not only the largest gradient
component.

We vary one parameter at a time to identify the limiting error. If the error
is unchanged when $m_f$ and $m_a$ increase, quadrature is limiting the
accuracy. If it is unchanged when $n_q$ increases, one of the Krylov spaces is
limiting the accuracy. Near a stationary input, the forward action may be
accurate while the coefficient gradient is still unresolved. We must
therefore check the adjoint path and the coefficient gradient as well. These
comparisons are refinement diagnostics rather than error bounds. For a
general nonnormal operator, the bound in Section~\ref{sec:gradient_error}
accounts for amplification by $e^{t\bA}$, which an Arnoldi residual alone does
not measure.

The required quadrature order depends on the problem. Eight nodes are
insufficient for the stiff Fokker--Planck example, while the Lindblad sweep
reaches $10^{-16}$ with 12 nodes. We therefore use a short refinement sweep
rather than a universal value of $n_q$.

\begin{figure}[tbp]
\centering\includegraphics[width=\linewidth]{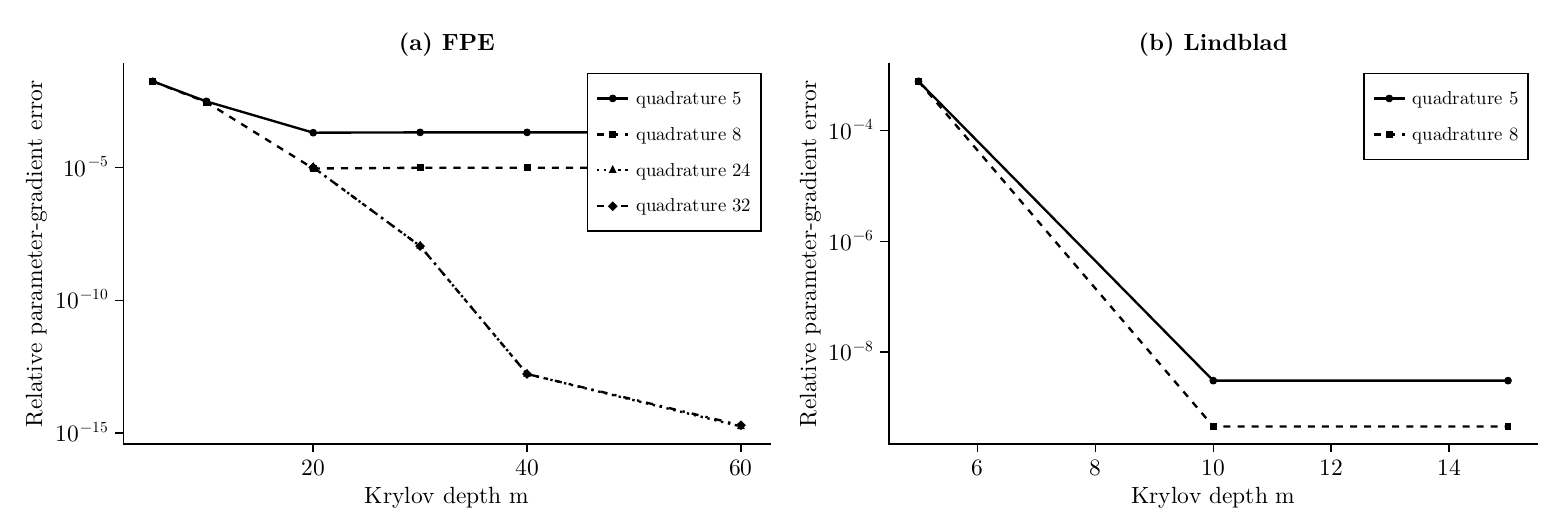}
\caption{Measured Krylov-depth sweeps for Fokker--Planck ($h=0.1$) and
Lindblad ($h=1$) operators, against dense block-Fr\'{e}chet references.
Line styles and markers identify quadrature orders. Refining the basis alone
cannot remove a quadrature floor. We report the higher-order Fokker--Planck
values in the supplement.}
\label{fig:accuracy_bw}
\end{figure}
\paragraph{Accuracy after a coefficient map.}
If another model computes the coefficients $\btheta$, the chain rule
multiplies the coefficient gradient by the Jacobian of that model. This
Jacobian may also amplify an error in the gradient. A tolerance used for
optimization should therefore be checked after the chain rule. The supplement
examines sensitivity-scaled coordinates and the neural rate map.

\subsection{Library reuse at matched accuracy}
\label{sec:library_accuracy}
To isolate the effect of library size, we hold a 320-state birth--death
operator fixed and increase the number of distinct two-entry directions from
16 to 512. For each library, we select the fastest tested configuration whose
forward, parameter-gradient, and input-state-gradient errors all satisfy the
target tolerance. At tolerance $10^{-8}$, separate directional block
exponentials require between 4.37 and 131 ms as the library grows. Reusing the
two Krylov spaces requires between 0.884 and 3.48 ms, giving a $38\times$
speedup at 512 directions. Every selected configuration has an error below
$8.3\times10^{-14}$. The supplement gives the complete sweep and the
fused-contraction results.

\subsection{Comparison with differentiated Arnoldi at matched accuracy}
\label{sec:matfree}

We compare our calculation with reverse differentiation of a matrix-free
exponential routine in Matfree \cite{kraemer2024matfree}. Matfree differentiates
the finite Arnoldi calculation. Our method first writes the derivative of the
exact finite-time solution and then approximates the forward and adjoint
paths. Both calculations converge to the same gradient as their Krylov spaces
are refined. For each state dimension, we report the fastest configuration
that reaches the prescribed accuracy.

Table~\ref{tab:matfree} reports median wall-clock times on an NVIDIA RTX 4070
SUPER. Both methods are compiled with JAX in the same process, which avoids a
cross-language comparison. At tolerance $10^{-6}$, Matfree is faster for
$n\leq320$, while our method is faster from $n=10^{3}$. At tolerances
$10^{-8}$ and $10^{-10}$, the crossover occurs at $n=10^{4}$. The relative
advantage grows with the state dimension and decreases as the tolerance is
tightened. In sparse-library problems, the state space and the number of
candidate directions may both be large, while sampling noise often sets the
required accuracy above machine precision.

\begin{table}[!tbp]
\centering
\caption{Comparison with Matfree at fixed accuracy on an NVIDIA RTX 4070 SUPER.
Entries are median times in milliseconds. Parentheses contain the achieved
maximum relative error over the parameter and input-state gradients for the
fastest tested configuration that meets the column
tolerance. There are 31 evaluations after compilation and warmup. Both methods
are JIT-compiled in the same process. Bold marks the faster method.}
\label{tab:matfree}
{\small
\resizebox{\linewidth}{!}{\begin{tabular}{rrrrrrr}
\toprule
& \multicolumn{2}{c}{tol $10^{-6}$} & \multicolumn{2}{c}{tol $10^{-8}$} & \multicolumn{2}{c}{tol $10^{-10}$}\\
\cmidrule(lr){2-3}\cmidrule(lr){4-5}\cmidrule(lr){6-7}
$n$ & ours & Matfree & ours & Matfree & ours & Matfree \\
\midrule
80      & 4.24 ($3.0\mathrm e{-9}$) & \textbf{2.53} ($2.5\mathrm e{-8}$) & 4.24 ($3.0\mathrm e{-9}$) & \textbf{3.06} ($1.3\mathrm e{-12}$) & 5.43 ($2.1\mathrm e{-13}$) & \textbf{3.06} ($1.3\mathrm e{-12}$) \\
320     & 3.30 ($8.0\mathrm e{-7}$) & \textbf{3.00} ($5.4\mathrm e{-8}$) & 4.57 ($3.6\mathrm e{-10}$) & \textbf{3.65} ($3.4\mathrm e{-12}$) & 5.76 ($2.6\mathrm e{-11}$) & \textbf{3.65} ($3.4\mathrm e{-12}$) \\
1{,}000 & \textbf{3.81} ($3.2\mathrm e{-7}$) & 4.36 ($7.4\mathrm e{-8}$) & 5.86 ($1.0\mathrm e{-10}$) & \textbf{5.34} ($4.6\mathrm e{-12}$) & 5.94 ($1.0\mathrm e{-10}$) & \textbf{5.34} ($4.6\mathrm e{-12}$) \\
3{,}000 & \textbf{5.69} ($5.6\mathrm e{-7}$) & 6.37 ($4.9\mathrm e{-7}$) & 8.51 ($1.7\mathrm e{-10}$) & \textbf{8.36} ($3.0\mathrm e{-11}$) & 11.09 ($2.7\mathrm e{-11}$) & \textbf{8.36} ($3.0\mathrm e{-11}$) \\
10{,}000& \textbf{4.21} ($8.8\mathrm e{-7}$) & 9.99 ($1.8\mathrm e{-10}$) & \textbf{8.30} ($2.6\mathrm e{-10}$) & 9.99 ($1.8\mathrm e{-10}$) & \textbf{10.42} ($4.3\mathrm e{-11}$) & 11.85 ($3.1\mathrm e{-12}$) \\
\bottomrule
\end{tabular}}
}
\end{table}

The timing sweep uses Float64, $M_{\rm lib}=1000$, $h=0.3$, and a local
birth--death generator with $3n-2$ nonzeros. Accuracy is the maximum of
parameter- and input-state-gradient relative errors against an independently
refined exponential reference. Matching parameter gradients alone would not
reproduce this timing comparison. The supplement adds the complete timing
curves, dense automatic differentiation, the two-sided projection, CPU
results, and additional sparse-contraction comparisons.

The forward and adjoint paths may require different depths. A
linear quantity $\bm c^*\bm u$ is invariant when $\bm c^*\bA=0$. If also
$\bm c^*\bA_j=0$, its derivative with respect to coefficient $j$ must vanish.
The exact gradient formula preserves this property for CTMC mass and Lindblad
trace. Differentiating an under-resolved Arnoldi calculation may give a
nonzero derivative unless the invariant is projected explicitly. The two
methods agree after refinement or explicit invariant projection.

The forward path, adjoint path, and quadrature require separate accuracy
checks. We refine them in that order and test both the input-state and
coefficient gradients. The supplement gives under-resolved and nearly
stationary examples in which checking only the forward value is insufficient.

\section{Applications}
\label{sec:stls}

We apply the method to three inverse problems. The Brusselator CME uses an
overcomplete reaction library with rates that vary across experimental
conditions. The Kuramoto--Sivashinsky model places a nonlinear calculation
before the exponential and therefore requires the input-state derivative.
The Lindblad model has a complex generator with real coefficients. The
forward and adjoint Krylov spaces are constructed in the same way for all
three problems. Only the coefficient map and the optimization procedure
change.

\subsection{Sparse optimization}
\label{sec:overview}

We perform thresholding in the scaled coordinates $\bphi=\bD\btheta$. The
$j$th diagonal entry of $\bD$ is the norm of the initial sensitivity for
coefficient $j$. The threshold therefore compares changes in the model output
rather than coefficients that may have different physical units.

Each smooth optimization is restricted to the current support and uses
Algorithm~\ref{alg:shared_arnoldi} for its gradient. After the optimization,
we project negative coefficients to zero and remove coefficients below a
fixed fraction of the largest scaled coefficient. When the support stops
changing, we solve once more without thresholding. The supplement gives the
complete pseudocode. In the Lindblad experiment, every candidate receives the
same initial value and the algorithm is given no support information. The
scaling and thresholding belong to the outer optimization, not to the
gradient calculation.

\subsection{Brusselator CME (ungrouped finite-time identification)}
\label{sec:snapshot_controls}
\label{sec:weak_comparator}
\label{sec:experiments}
For a reaction network with copy-number state $\bx\in\N^d$, reaction $m$ has
stoichiometric change $\bzeta_m$. On a finite state space
$\calX=\{\bx_1,\ldots,\bx_n\}$, the probability vector obeys
$\dot\bp=\bA\bp$. A mass-action candidate consuming
$\bnu=(\nu_1,\ldots,\nu_d)$ molecules uses
\begin{equation}
\label{eq:mass_action_feature}
 \phi_{\bnu}(\bx)=\prod_{i=1}^d\binom{x_i}{\nu_i}.
\end{equation}
For each proposed stoichiometry $\bzeta_m$ and feature $\phi_k$, the fixed
gain--loss direction is defined with
$\chi_m(\bx)=\mathbf 1_{\{\bx+\bzeta_m\in\calX\}}$ as
\begin{equation}
\label{eq:cme_library_direction}
 [\bA_m^{(k)}]_{\by,\bx}
 =\chi_m(\bx)\phi_k(\bx)\bigl(\mathbf 1_{\{\by=\bx+\bzeta_m\}}
                   -\mathbf 1_{\{\by=\bx\}}\bigr),
 \qquad \bx,\by\in\calX.
\end{equation}
A reaction that would leave $\calX$ is suppressed, while every retained gain
has a matching diagonal loss. This boundary convention keeps probability in
the truncated state space rather than allowing it to leak through the
boundary. The complete generator is
\begin{equation}
\label{eq:cme_affine_library}
 \bA(\btheta)=\sum_{m=1}^{M}\sum_{k=1}^{K}
 \theta_{m,k}\bA_m^{(k)},\qquad \theta_{m,k}\geq0.
\end{equation}
Each candidate mechanism defines one sparse matrix, and the support of
$\btheta$ determines the inferred reaction network. The conservative
truncation gives $\mathbf1^T\bA_m^{(k)}=0$ for every direction.

We consider the Brusselator on a $13\times13$ count grid, so $n=169$. The
candidate library contains all 40 pairings of four stoichiometric channels
with bivariate mass-action monomials of total degree at most three. The active
mechanisms are
$\varnothing\to X$, $X\to\varnothing$, $X\to Y$, and
$2X+Y\to3X$ (library directions 1, 13, 23, and 39). We do not use channel
grouping or assume a support size. We fit four experimental conditions, each
with six initial distributions and exact noiseless snapshots separated by the
nonuniform gaps $0.05$, $0.20$, and $0.50$. Four additional conditions and a
time-$1.5$ prediction are reserved for testing.

The nonnegative linear rate model is
$\theta_j(A,B)=W_{j,:}[1,A,B]^T$. Adam first minimizes the finite-time
snapshot loss with an $\ell_1$ penalty over all 40 rows. A fixed coefficient
threshold then selects the support. The selected model is refitted with
unpenalized Adam followed by a limited-memory quasi-Newton (L-BFGS) solve. The
threshold is calibrated on seed 11 and then held fixed for the remaining four
seeds. All five runs recover exactly the four active directions, with no false
positives. The mean coefficient, test-transition, and long-horizon relative
errors are $4.27\times10^{-6}$, $4.54\times10^{-6}$, and
$3.99\times10^{-6}$, respectively.
Figure~\ref{fig:brusselator_headliner}(a) reports the small seed-to-seed
variation. The fit uses only snapshot distributions and no reaction event
counts.

The two weak baselines replace the finite-time transition by an integrated
form of the differential equation. For a smooth test function $\psi$ that
vanishes at the ends of its interval, weighted weak-form SINDy uses
integration by parts to fit
\begin{equation}
 -\int \psi'(t)\bp(t)\,dt
 =\sum_j\theta_j\int\psi(t)\bA_j\bp(t)\,dt.
 \label{eq:weak_cme_comparator}
\end{equation}
We evaluate these integrals from a piecewise-linear interpolant of the
snapshots. The interval rule instead applies the trapezoidal approximation on
every observed interval:
\begin{equation}
 \bp(t_{i+1})-\bp(t_i)
 \approx \Delta t_i\sum_j\theta_j\bA_j
 \frac{\bp(t_i)+\bp(t_{i+1})}{2}.
 \label{eq:interval_weak_comparator}
\end{equation}
Figure~\ref{fig:brusselator_headliner}(b) compares the finite-time Krylov fit
with these two baselines \cite{messenger2021weak}. All three methods use the
same noiseless distributions on the $13\times13$ state grid. They also use the
same training and test conditions and the same 40-direction library. At the
four nonuniform snapshots, the finite-time fit recovers exactly 4/4 directions
in every run, with mean test coefficient
error $4.27\times10^{-6}$. Weak SINDy recovers all four directions but adds two
false positives, with error $0.855$. The interval rule selects one false
positive and gives error $0.548$. For the weak fits, hyperparameters are
selected by comparison with the known coefficients. We repeat the comparison
with 6, 16, and 76 snapshots. The finite-time fit recovers the exact support
throughout this sweep after sampling-specific sparsity tuning. Its coefficient
errors at 4, 6, 16, and 76
snapshots are $1.78\times10^{-6}$, $2.33\times10^{-6}$,
$3.72\times10^{-5}$ and $7.79\times10^{-5}$. Weak SINDy becomes slightly more
accurate at 16 snapshots and is about four times more accurate at 76. The
finite-time fit is substantially more accurate at four and six snapshots,
where it avoids temporal quadrature across the widest observation gaps. The
supplement includes an all-edge 870-direction CTMC, nearly collinear
candidates, and omitted-mechanism experiments.

\begin{figure}[tbp]
\centering\includegraphics[width=\linewidth]{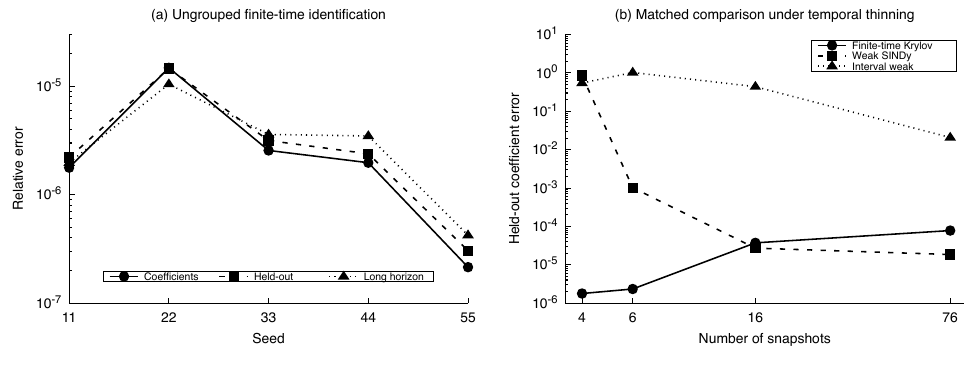}
\caption{Ungrouped snapshot-only Brusselator identification. (a) Coefficient,
test-transition, and time-1.5 prediction errors for five finite-time fits.
Every run selects exactly 4/40 directions. (b) Matched test coefficient
error for the finite-time Krylov fit and two weak-form baselines on the same
distributions and library. Four denotes the nonuniform snapshots used in panel
(a). The remaining points refine the observation grid. The finite-time curve
uses seed 11, 500 discovery epochs, 300 refit epochs and at most three randomly sampled
transitions per update, followed by 100 L-BFGS iterations on the final support.
The weak-form hyperparameters are selected by comparison
with the known coefficients. The finite-time threshold is fixed at $0.2$. Its
$\ell_1$ penalty is $10^{-4}$ except at 76 snapshots, where we reduce it to
$10^{-6}$ to account for the shorter intervals.}
\label{fig:brusselator_headliner}
\end{figure}
\FloatBarrier

\subsection{Kuramoto--Sivashinsky PDE}
\label{sec:ks_exp}
The KS example tests the input-state derivative.
\begin{equation}\label{eq:ks_pde}
 u_t=-a u_{xx}-b u_{xxxx}-\frac{c}{2}\partial_x(u^2),
\end{equation}
on a periodic domain of length 22. On 64 grid points, centered difference
matrices give
\begin{equation}\label{eq:ks_split}
 \bA(a,b)=-aD_{xx}-bD_{xxxx},\qquad
 \bm N_c(\bm u)=-\frac{c}{2}D_x(\bm u\odot\bm u),
 \qquad D_{xxxx}=D_{xx}^2.
\end{equation}
Here $\odot$ denotes componentwise multiplication.
We first take an explicit Euler step for the nonlinear term and then propagate
the result with the exponential of the linear operator:
\begin{equation}\label{eq:ks_expanded_step}
 \bv_k=\bm u_k-\frac{hc}{2}D_x(\bm u_k\odot\bm u_k),
 \qquad
 \bm u_{k+1}=e^{h\bA(a,b)}\bv_k.
\end{equation}
The affine directions $-D_{xx}$ and $-D_{xxxx}$ give the derivatives with
respect to $a$ and $b$ through Proposition~\ref{prop:adjoint_gradient}. The
coefficient $c$ enters through $\bv_k$. We first compute the derivative of the
loss with respect to $\bv_k$ using~\eqref{eq:composable_pullback}, and then
differentiate the nonlinear term. The supplement gives the complete chain
rule.

We take $h=0.0125$ and use Krylov depth 40 with eight quadrature nodes. A
fourth-order Runge--Kutta (RK4) solution with 20 substeps serves as the
reference. We fit ten transitions from each of three fields and test a
20-step prediction on a fourth field. Across five initializations, the mean
coefficient and test-path errors are $1.84\%$ and $0.0908\%$. If the
input-state derivative $\nabla_{\bv_k}\ell$ is omitted, the initial gradient
of $c$ is zero and the two errors increase to $18.6\%$ and $1.04\%$.

We then enlarge the model to 15 directions,
\begin{equation}\label{eq:ks_sparse_library}
 \{-D_x,-D_{xx},D_xD_{xx},-D_{xxxx}\}
 \ \cup\
 \left\{-\frac{1}{q}D_x(\bm u^{\odot q}):q=2,\ldots,12\right\}.
\end{equation}
Only $-D_{xx}$, $-D_{xxxx}$, and the quadratic flux are active. We add Gaussian
snapshot noise at $0$, $0.25\%$, $0.5\%$, and $1\%$ root-mean-square (RMS)
amplitude. Prediction accuracy is the relative $L^2$ error at $T=0.5$ on a
fourth, noise-free path that is not used for fitting. Seed 11 at $0.5\%$
fixes the $10^{-7}$ penalty and $0.08$ threshold for all 20 initializations and
all four levels. We recover the exact support in all 80 fits, with mean
active--inactive separation between $0.910$ and $0.957$.

Noise affects support recovery and coefficient estimation differently. At
$0.5\%$, refitting on the selected support reduces the mean full-vector coefficient error from
$4.22\%$ to $3.21\%$ and the noise-free test-path error from $0.407\%$
to $0.255\%$.
At $1\%$, it increases these errors from $5.98\%$ to $8.79\%$ and from
$0.508\%$ to $0.713\%$, respectively. The support therefore remains stable
even after unpenalized refitting ceases to improve the coefficient estimates.

To measure scaling, we use a length-704 domain with 2048 points and
Fourier-modulated second- and fourth-derivative libraries of size 16, 64, and
256. Reusing the two Krylov spaces takes $4.07$, $7.99$, and $24.1$ ms.
Rebuilding both spaces for every direction takes $46.0$, $188$, and $757$ ms.
The speedup ranges from $11.3$ to $31.5$, and the two calculations agree to
machine precision. The timings are minima of five warmed runs at depth 40
with eight quadrature nodes.

\begin{figure}[H]
\centering
\includegraphics[width=\linewidth]{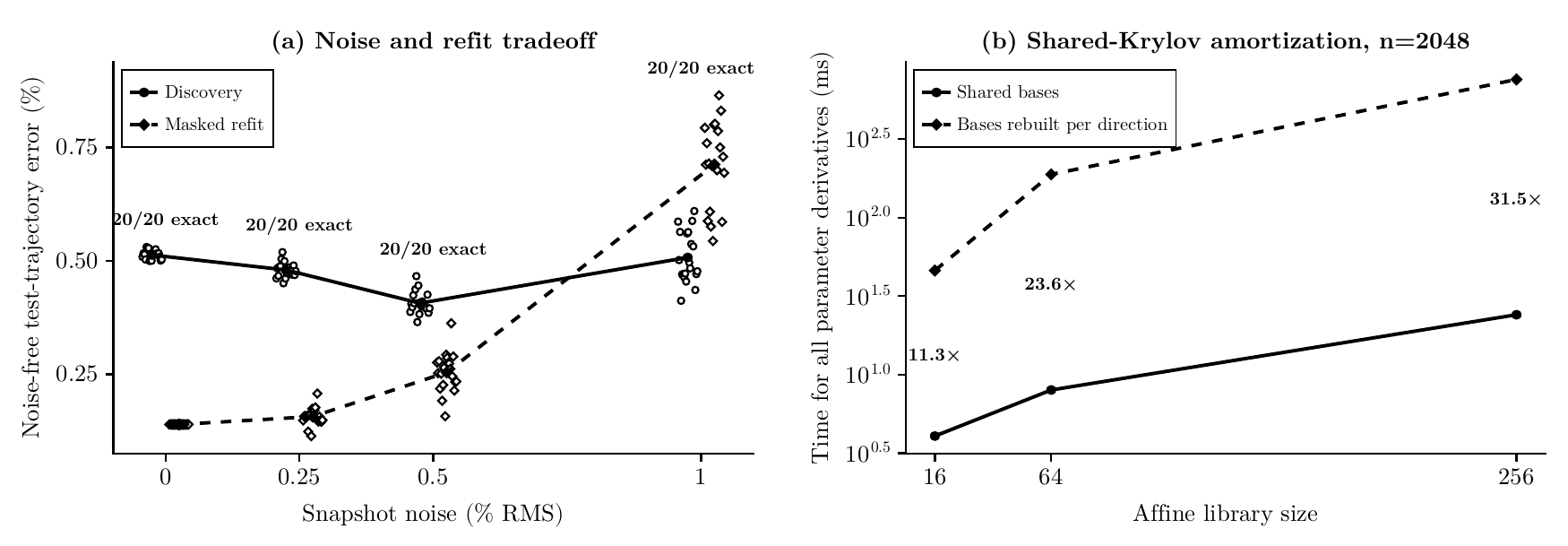}
\caption{Kuramoto--Sivashinsky robustness and computational scaling.
(a) Noise-free test-path errors over 20 initializations at each snapshot-noise
level. We recover the exact support in every initial sparse fit. The solid and dashed
curves show the means before and after refitting on the selected support.
(b) Time required to compute all parameter derivatives on the 2048-state
problem when the two Krylov spaces are reused or rebuilt for every affine
direction. We annotate the curves with the measured speedups.}
\label{fig:ks_identification}
\end{figure}
\FloatBarrier

\subsection{Lindblad equation}
\label{sec:lindblad_exp}

For an open quantum system, let
$H(\bm J)=H_0+\sum_jJ_jH_j$ and let $L_k$ denote candidate jump operators.
The density matrix satisfies
\begin{equation}
\label{eq:lindblad_structured_operator}
\dot\rho=-i[H(\bm J),\rho]
 +\sum_k\gamma_k\left(
 L_k\rho L_k^*-\tfrac12\{L_k^*L_k,\rho\}\right).
\end{equation}
Here $[X,Y]=XY-YX$ is the commutator,
$\{X,Y\}=XY+YX$ is the anticommutator, and $X^*$ denotes the conjugate
transpose. The commutator describes Hamiltonian evolution, while each term in
the sum is a dissipative channel. Let
$\bm{\rho}_{\rm v}=\operatorname{vec}(\rho)$. Using
$\operatorname{vec}(AXB)=(B^T\otimes A)\operatorname{vec}(X)$, we obtain
\begin{equation}
\label{eq:lindblad_kronecker_library}
 \dot{\bm{\rho}}_{\rm v}=\left(\mathcal L_0+
 \sum_jJ_j\mathcal H_j+\sum_k\gamma_k\mathcal D_k\right)\bm{\rho}_{\rm v},
 \qquad
 \mathcal H_j=-i\left(I\otimes H_j-H_j^T\otimes I\right),
\end{equation}
where, with $K_k=L_k^*L_k$,
\begin{equation}
\label{eq:lindblad_dissipative_direction}
 \mathcal D_k=\overline{L_k}\otimes L_k
 -\frac12\left(I\otimes K_k+K_k^T\otimes I\right).
\end{equation}
The Hamiltonian commutators and dissipators are fixed directions on
$\mathbb C^{n^2}$ with real coefficients $(\bm J,\bm\gamma)$. The generator
is complex, and trace replaces probability mass as the conserved linear
quantity. We therefore use the conjugate-transpose adjoint in
\eqref{eq:general_vjp}. The resulting coefficient gradient is real and is used
in the thresholded optimization.

For $q$ qubits, the vectorized density state has dimension $4^q$. We fix a
transverse field as $\bA_0$ and construct an overcomplete library. It contains
$XX+YY$ exchange couplings on all $\binom{q}{2}$ site pairs, together with
dephasing, decay, and bit-flip dissipators on every site. Thus,
$M_{\rm lib}=\binom{q}{2}+3q$. Only nearest-neighbor exchange and $\sigma_z$
dephasing are active in the true generator. We initialize every candidate at
the same nonzero value and infer the $2q-1$ active directions without using
the sparsity pattern.

Table~\ref{tab:lindblad_main} reports five independent rate draws at each
size. The supplement gives the individual runs, partial-observation fits, and
finite-shot experiments.

\begin{table}[tbp]
\centering
\caption{Sparse Lindblad identification over five rate draws per size. Every
run recovers the exact support, with no false positives or false negatives.}
\label{tab:lindblad_main}
\begin{tabular}{rrrrr}
\toprule
Qubits & Liouville dimension & Library & Active & Mean rate error \\
\midrule
2 & 16   & 7  & 3 & $7.5\times10^{-9}$ \\
3 & 64   & 12 & 5 & $5.6\times10^{-9}$ \\
4 & 256  & 18 & 7 & $1.3\times10^{-8}$ \\
5 & 1024 & 25 & 9 & $1.4\times10^{-8}$ \\
6 & 4096 & 33 & 11 & $7.3\times10^{-9}$ \\
\bottomrule
\end{tabular}
\end{table}

The larger Lindblad experiments test the sparse contraction routine rather
than support recovery. The supplement reports this scaling separately for
seven through eleven qubits.

\section{Discussion and Conclusions}
\label{sec:conclusion}
We have developed a method for parameter inference in large, sparse linear
ODEs whose generators depend affinely on unknown coefficients. For each
exponential action, one forward Krylov space and one adjoint Krylov space are
used for the entire coefficient library. Gauss--Legendre quadrature reduces
the coefficient gradient to the contractions
$\bm\lambda^*\bA_j\bm f$. The same computation gives derivatives with respect
to the input state and interval length, which allows the exponential action
to be used inside a larger model.

The error analysis separates quadrature error from errors in the forward and
adjoint paths. Its additive form remains sensitive when only one path is
unresolved, and the contractive bounds for CTMC and Lindblad generators use
residual quantities computed by Arnoldi. The numerical examples show why
this separation matters. A residual-product estimate can vanish after
one-sided termination, while successive differences can become small during
slow convergence. We therefore use the residual bounds as safeguards and
check successive depths and quadrature orders separately.

The experiments use the same gradient calculation for a CME, a semilinear
PDE, and a Lindblad equation. They recover sparse mechanisms from finite-time
data, differentiate through a nonlinear step, and treat complex generators.
At matched accuracy, reusing the two Krylov spaces reduces the cost when the
state space and coefficient library are large.
Quadrature also permits a one-workspace implementation that replaces one
full Krylov basis by the forward path at the quadrature nodes. In the
12-qubit Lindblad test, this schedule fits within 16~GB. A projected-core
calculation reaches the same memory scale only after rebuilding one basis.

The present experiments are synthetic, and the main identification tests use
full-state observations. The method removes repeated work across the library
but does not remove growth in the state dimension, such as the $4^q$ density
state for $q$ qubits. Compressed state representations and adaptive Krylov
solvers are natural directions for larger problems and longer time
intervals.

\paragraph{Code availability}
Source code and experiment scripts are available at
\url{https://github.com/AdityaDendukuri/AdjointKrylovGradients}.

\fi
\ifsupplementcontent
\ifdefined\buildcombined
\clearpage
\begin{center}
{\Large\bfseries Supplementary Material}
\end{center}
\fi
\appendix
The supplement begins with the context-dependent rate models used in the
applications. We then give additional accuracy and cost results, followed by
the CME construction, inference details, and numerical checks. The notation
is the same as in the main paper.

\section{Neural Context-to-Rate Learning}
\label{sec:neural}
This section extends the constant-rate model to coefficients that vary with
experimental context while retaining a common support.
\subsection{Context-dependent rates}
\label{sec:context_rates}

Let $\xi$ denote an experimental condition, such as an input concentration.
We keep the candidate library fixed across conditions and write
\begin{equation}
\label{eq:context_generator}
\btheta(\xi)=\bm s\odot\bm r(\xi;\bw),\qquad
\bA(\xi)=\bA_0+
\sum_{j=1}^{M_{\rm lib}}\theta_j(\xi)\bA_j.
\end{equation}
where $\odot$ denotes componentwise multiplication. The trainable weights
$\bw$ determine the rate map, and $s_j\in\{0,1\}$ is a support mask shared by
all contexts. For a CME library, the parameterization or optimization
constraints enforce $r_j(\xi;\bw)\geq0$. The rate map may contain fitted
constants, a linear context model, or a neural network with nonnegative
outputs. The mask selects the possible mechanisms, and the rate map determines
how their rates change with context. A selected mechanism may still have zero
rate under a particular condition.

For context $\xi_c$, let $\hat\bp_{c,t}$ denote the observed full-state
snapshot and $\Delta t_{c,t}$ the next observation interval. Holding the
context fixed within each interval, we extend~\eqref{eq:loss} as
\begin{equation}
\label{eq:context_loss}
\calL_{\rm ctx}(\bw;\bm{s})=
\frac12\sum_{c=1}^{N_c}\sum_{t=0}^{T_c-1}
\left\|e^{\bA(\xi_c)\Delta t_{c,t}}\hat\bp_{c,t}
-\hat\bp_{c,t+1}\right\|_2^2.
\end{equation}
Here $N_c$ is the number of training contexts. Each additional initial
condition contributes another snapshot sequence to the loss. Support
selection uses only the training contexts and assumes that the possible
mechanisms are common to all of them, although their rates may differ.

\subsection{Gradients through a context-to-rate map}
\label{sec:context_gradient}

For a fixed support mask, the adjoint Krylov calculation computes the gradient
with respect to the rates $\btheta(\xi_c)$. Applying the chain rule gives
\begin{equation}
\label{eq:context_chain}
\nabla_{\bw}\calL_{\rm ctx}
=\sum_{c=1}^{N_c}
\bigl(D_{\bw}\bm r(\xi_c;\bw)\bigr)^T
\left(\bm{s}\odot\nabla_{\btheta(\xi_c)}\calL_{\rm ctx}\right).
\end{equation}
For a neural rate map, reverse differentiation applies the network Jacobian
to the coefficient gradient. A constant-rate model uses the coefficient
gradient directly. Each observed transition has its own forward and adjoint
Krylov spaces. We reuse these spaces across the library directions for that
transition, but not across contexts.

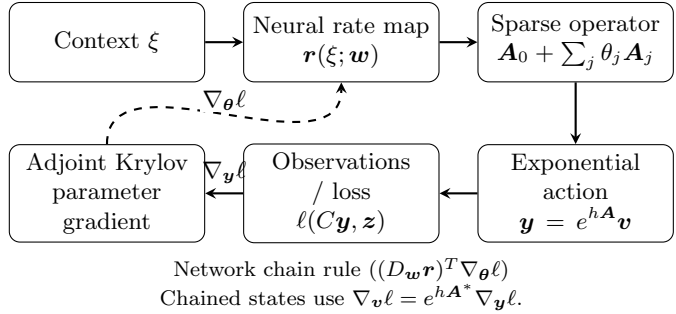
\begin{figure}[tbp]
\centering
\begin{tikzpicture}[>=stealth, font=\small,
box/.style={draw,rounded corners,align=center,text width=2.35cm,minimum height=1.05cm},
arr/.style={->,thick}]
\node[box] (c) at (0,0) {Context $\xi$};
\node[box] (nn) at (3.1,0) {Neural rate map\\$\bm r(\xi;\bw)$};
\node[box] (a) at (6.2,0) {Sparse operator\\$\bA_0+\sum_j\theta_j\bA_j$};
\node[box] (e) at (6.2,-2) {Exponential action\\$\bm y=e^{h\bA}\bv$};
\node[box] (loss) at (3.1,-2) {Observations / loss\\$\ell(C\bm y,\bm z)$};
\node[box] (vjp) at (0,-2) {Adjoint Krylov\\parameter gradient};
\draw[arr] (c)--(nn);\draw[arr] (nn)--(a);\draw[arr] (a)--(e);
\draw[arr] (e)--(loss);\draw[arr] (loss)--node[above]{$\nabla_{\bm y}\ell$}(vjp);
\draw[arr,dashed] (vjp.north) to[out=90,in=270] node[above,fill=white,inner sep=1pt]{$\nabla_{\btheta}\ell$}(nn.south);
\node[align=center,font=\footnotesize] at (3.1,-3.2)
{Network chain rule ($(D_{\bw}\bm r)^T\nabla_{\btheta}\ell$)\\
Chained states use $\nabla_{\bv}\ell=e^{h\bA^*}\nabla_{\bm y}\ell$.};
\end{tikzpicture}
\caption{Gradient calculation for a context-dependent rate model. Solid
arrows show the forward model and the derivative of the loss. The dashed
arrow passes the coefficient gradient to the network. The optional support
mask is separate from the network weights.}
\label{fig:neural_pipeline}
\end{figure}

\subsection{Brusselator (grouped neural rates and an ungrouped comparison)}
For the neural Brusselator model, we use a $41\times41$ count grid
($n=1681$), a 40-direction stoichiometry--propensity library, and the context
$\xi=(A,B)$. The library is divided into four known stoichiometric channels,
each with ten candidate propensity functions. Within a channel, a softmax map
produces nonnegative weights that sum to one. We gradually lower its
temperature so that the weights concentrate on a small number of candidates.
A neural network maps the context to the four channel amplitudes, and a second
stage refits their rates after selection. We use two complementary initial
conditions to distinguish the candidate responses. Because the architecture
imposes the four channel groups, selection within these groups is not an
unrestricted library search. The fit selects the four true directions. On
the test contexts, the mean relative errors in $A$, $B$, $k_2$, and $k_4$ are
$1.2\%$, $0.4\%$, $2.4\%$, and $2.8\%$, respectively
(Figure~\ref{fig:neural_results}).

For comparison with context-aware reactive SINDy, we regress trajectory
increments on instantaneous propensities with a nonnegative elastic-net fit.
Across three seeds, this fit retains all four true directions together with
$4\pm1$ false directions. A variant with shared coefficients retains
$6\pm1$ false directions. These methods use trajectory increments, whereas
the neural exponential model fits histogram snapshots. The comparison is
therefore between complete inference procedures and does not isolate the
gradient calculation.

We also test an ungrouped linear context map on a smaller noiseless problem.
All five runs recover the four active directions, with a mean coefficient
error of $0.875\%$ and a test-path error of $1.11\%$.
Section~\ref{sec:context_experiment} gives the calibration and confirmation
protocol. This experiment tests support selection without channel
constraints, but its linear rate map does not test a neural architecture.

\subsection{Stuart--Landau diffusion across a bifurcation}
On a $30\times30$ grid ($n=900$), we discretize
$\partial_t p=-\nabla\cdot(\bm b p)+D\Delta p$ with $D=0.15$ and
\[
\bm b(x,y;\mu)=
\begin{pmatrix}\mu x-y-(x^2+y^2)x\\x+\mu y-(x^2+y^2)y\end{pmatrix}.
\]
A degree-three drift library contains 20 directions, of which eight are
active. A neural map with two width-32 $\tanh$ hidden layers learns their
coefficients across eight contexts $\mu\in[-0.3,1.1]$. We start from an
off-center Gaussian so that the observations contain rotational motion. A
centered radial density would not contain this information. With an
observation interval of one, the initial sparse fit and the subsequent refit
recover all eight active terms. For seed 42, the mean test-context errors are
approximately $1.5\%$ for the coefficient determined by $\mu$ and $0.5\%$ for
the rotation coefficient.

An RK4 comparison with 32 substeps selects two terms, only one of which is
true. It uses 100 discovery epochs and takes 294 seconds. The exponential
model uses 800 discovery epochs and takes 92 seconds, so these timings do not
compare equal optimization budgets. We also train unstructured neural ODEs
with hidden widths 16 and 64, corresponding to 30,004 and 120,388 parameters.
Both models use 32 RK4 substeps per unit time and 800 epochs. Across three
seeds, their mean extrapolation losses at $t=6$ are
$6.27\times10^{-3}$ and $5.88\times10^{-3}$. Negative values occur in
$44.5\%$ and $43.6\%$ of the grid cells, and the mean mass errors are $0.030$
and $0.050$. Further details, including the support fit at a shorter
observation interval, appear below.

\begin{figure}[tbp]
\centering\includegraphics[width=\linewidth]{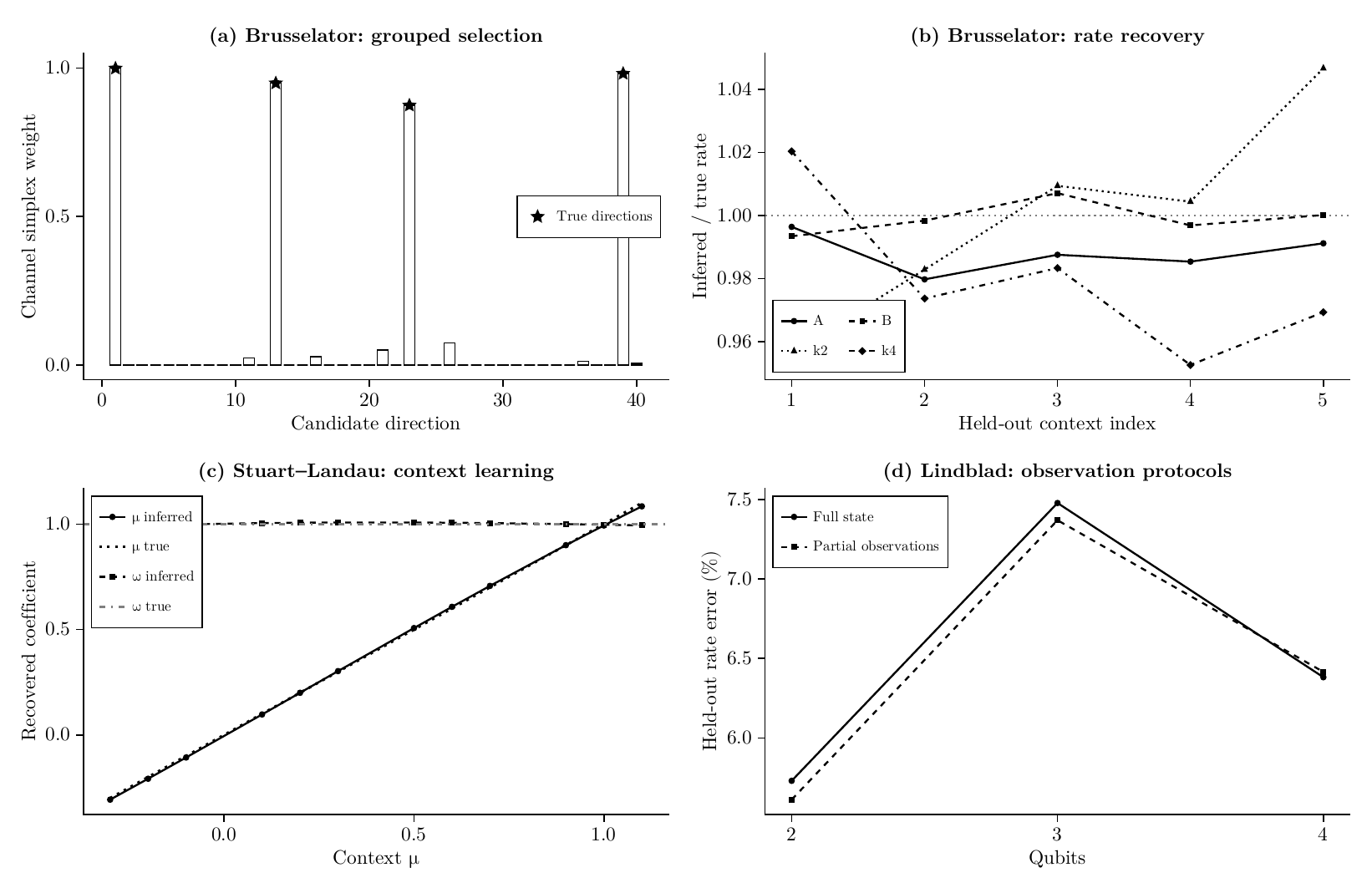}
\caption{Context-dependent rate experiments. Black lines, markers, and line
styles distinguish the recorded runs. (a) Grouped Brusselator simplex
weights, with the true terms marked by stars. (b) Brusselator rate ratios on
the test contexts. (c) Stuart--Landau coefficient functions. (d) Lindblad
rate errors under the two known-support observation protocols.}
\label{fig:neural_results}
\end{figure}
\subsection{Neural Lindblad rates and partial observations}
\label{sec:lindblad_neural}

A separate experiment fixes the support of the $2q-1$ nearest-neighbor
exchange and dephasing terms and learns their rates as functions of a scalar
context $\xi$.
The fixed field in this experiment is $0.3\sum_i\sigma_z^{(i)}$.
The prescribed rates are
$J_i(\xi)=0.5[1+0.4\sin(i\pi\xi)]$ and
$\gamma_i(\xi)=0.1[1+0.5\cos^2(i\pi\xi)]$.
A network with two width-64 $\tanh$ hidden layers represents the rate map. A
softplus output $x\mapsto\log(1+e^x)$ and a fixed positive floor keep every
inferred rate positive. Ten training contexts
$0.05,0.15,\ldots,0.95$ and five test contexts
$0.10,0.30,\ldots,0.90$ use observations at times $0.3,1,3$.
Adam uses step size $0.003$ for 600 epochs at $q=2,3$ and 800 epochs at $q=4$.

Full-state training fits consecutive density snapshots. For partial
observations, we instead propagate the known initial state
$\rho_0=(|+\rangle\langle+|)^{\otimes q}$ directly to each observation time,
minimizing $\frac12\sum_t\|C e^{t\bA(\bm r(\xi;\bw))}
\operatorname{vec}(\rho_0)-\bm y_t\|^2$.
The rows of $C$ measure the identity and all single-site Pauli expectations.
They also measure nearest-neighbor $XX$, $YY$, and $ZZ$ expectations. The
observation map has $6q-2$ rows.
The partial-observation fit does not use intermediate full density matrices.
Thus the full-state and partial experiments differ in both their observation
maps and their losses. Table~\ref{tab:lindblad_context} reports one
initialization, seed 42, for each case. We average the componentwise relative
rate errors over the test contexts. These experiments evaluate neural rate
fitting on a fixed support under the two observation protocols.

\begin{table}[tbp]
\centering
\caption{Known-support neural Lindblad fitting. The partial data are noiseless
expectation values from a known initial state. Errors are measured in the
rates on the test contexts.}
\label{tab:lindblad_context}
\begin{tabular}{rrrrr}
\toprule
Qubits & Full-state dim. & Observables & Full-state error & Partial error \\
\midrule
2 & 16 & 10 & $5.73\%$ & $5.61\%$ \\
3 & 64 & 16 & $7.48\%$ & $7.37\%$ \\
4 & 256 & 22 & $6.38\%$ & $6.42\%$ \\
\bottomrule
\end{tabular}
\end{table}

With sparse operators, full-state neural rate fitting also reaches five and
six qubits. For seed 42 and 800 epochs, the mean test-context rate errors are
$8.77\%$ and $8.76\%$, and the training times are 256 and 930 seconds. These
are known-support fits of nine and eleven coefficients. They are separate
from the unknown-support experiments through six qubits and the gradient
scaling experiment through eleven qubits.

Across three seeds, local $Z$ observations alone produce a nearly zero
observation loss but mean rate errors of $293$--$338\%$. Adding single-site
and nearest-neighbor Pauli observations reduces the mean errors to $5.44\%$,
$7.93\%$, and $6.12\%$ at two, three, and four qubits. Simulating 100
measurement shots increases these errors to $11.5\%$, $16.3\%$, and
$15.9\%$. With 10,000 shots, the errors approach their noiseless values.
Thus a small observation loss from local $Z$ measurements need not imply
accurate rates. Results for the intermediate shot counts and smaller
observation sets appear below.

\section{Additional Discussion}
\label{app:additional_discussion}
The forward and adjoint Krylov spaces are shared by all coefficient
directions on an observation interval. Only the library contractions grow
with the number of directions. Derivatives with respect to the input state
and interval length connect the exponential action to the calculations that
precede and follow it. The cost analysis in the main paper identifies when
reuse of the two spaces reduces the runtime.

The applications vary the active coefficient set, the dependence of the
rates on context, and the observations entering the loss. Each observed
transition requires its own forward and adjoint spaces. In the nonlinear
ablation, we isolate the input-state derivative by omitting that term from the
chain rule. The CME, diffusion, and Lindblad examples include real and complex
generators with different conservation laws.

Reuse is most effective when construction of the two Krylov spaces costs more
than contraction with the operator library. Differentiated Arnoldi and dense
automatic differentiation are faster in other parts of the measured
accuracy--cost range. The forward approximation, adjoint approximation, and
quadrature must all be resolved before errors caused by support ambiguity,
observation noise, model mismatch, or time discretization can be interpreted.

Adaptive selection of the forward and adjoint depths remains useful when
library storage and the full training cost are included. Automatic
thresholding of context-dependent rate functions is a separate inference
problem. The experiments below measure gradient accuracy and record the
behavior of the stated fitting procedures on synthetic data.

\subsection{Failure modes of heuristic stopping estimates}
\label{app:stopping_failure}

Figure~\ref{fig:stopping_failure} shows two ways in which a heuristic stopping
estimate can become small before the approximation is accurate. Panel (a)
uses the symmetric generator of a reflecting random walk on 200 states, an
interval of length 80, the stationary initial vector
$\bv=200^{-1/2}(1,\ldots,1)^T$, and terminal vector
$\nabla_{\bm y}\ell=\bdelta_1$. The generator is negative semidefinite, so its
semigroup is contractive in the 2-norm. The forward Arnoldi process terminates
after one step, while the adjoint depth ranges from 1 to 30. The multiplicative
estimate is therefore zero at every plotted depth, even though the true
Frobenius error decreases only from $20.66$ to $0.137$.

The curve labeled ``additive residual bound'' controls the full Fr\'{e}chet
matrix rather than an individual gradient component. The unit-norm forward
path is exact, so the expansion used in
Theorem~\ref{thm:residual_gradient_error} gives
\[
 \|L_{\exp}(h\bA,(\nabla_{\bm y}\ell)\bv^*)-\widetilde L\|_F
 \leq |h|\int_0^1\|\bm e_\lambda(s)\|_2\,ds.
\]
Variation of constants and contractivity in the 2-norm bound the right-hand
side by the integrated adjoint Arnoldi residual. The resulting bound decreases
from $79.0$ to $0.527$ over the plotted depths.

Panel (b) uses the final interval, with $h=0.5$, of the stiff 1681-state
Brusselator problem and a 40-direction cubic library. The plotted
approximation and the depth-420 reference both use the two-sided projection,
so quadrature error is absent from this comparison. At depth 16, the
successive-approximation estimate of
\cite[Eq.~(60)]{kandolf2021lowrank} is $5.40\times10^{-9}$, whereas the
normalized library-gradient error is $6.40\times10^{-7}$. At depth 120, these
values are $2.08\times10^{-9}$ and $4.07\times10^{-8}$. The additive curve
uses the $1$-norm for the forward error and the $\infty$-norm for the adjoint
error in Corollary~\ref{cor:ctmc_residual_bound}. The CTMC semigroup is
nonexpansive in the first norm, and its adjoint is nonexpansive in the second.
Although the bound remains above the true error, its smallest value through
depth 120 is $7.32\times10^{-6}$ and does not certify the target $10^{-8}$.
For this problem, the bound detects unresolved error but does not determine an
efficient stopping depth.

\begin{figure}[tbp]
\centering
\includegraphics[width=\linewidth]{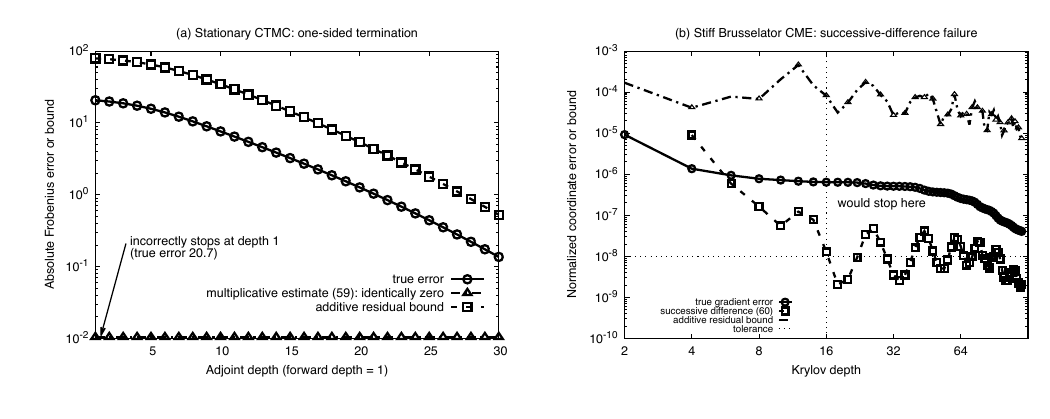}
\caption{Two failures of heuristic stopping estimates. (a) One-sided
termination makes the residual-product estimate of
\cite[Eq.~(59)]{kandolf2021lowrank} identically zero. (b) In the stiff
Brusselator problem, the difference between successive projected
approximations crosses $10^{-8}$ before the true normalized gradient error.
The additive residual bounds remain valid but conservative.}
\label{fig:stopping_failure}
\end{figure}

\subsection{Derivation of the complexity bound}
Proposition~\ref{prop:complexity} separates the construction of the
Krylov spaces, evaluation at the quadrature nodes, and contraction with the
library directions. In one Arnoldi process, the $m_K$ applications of $\bA$
or $\bA^*$ cost
$O(m_K\nnz(\bA))$. At step $r$, modified Gram--Schmidt orthogonalizes a
length-$n$ vector against $r$ basis vectors. Hence
\[
 O\left(n\sum_{r=1}^{m_K}r\right)=O(nm_K^2).
\]
Each exponential action requires a forward basis and an adjoint basis. Let
$N_{\rm act}$ denote the number of such actions. The first term of the bound
absorbs the constant factors associated with the two constructions.

At one quadrature node, the two reduced matrix exponentials cost
$O(m_K^3)$, and multiplication by the two bases costs $O(nm_K)$. Repeating
these operations over $n_q$ nodes and $N_{\rm act}$ actions gives the second
term. For a fixed action, node, and direction, the sparse product
$\bA_j\bm f_q$ and its inner product with $\bm\lambda_q$ cost
$O(\nnz(\bA_j)+n)$. Summing over directions, nodes, and actions gives
\[
 O\left(N_{\rm act}n_q\sum_{j=1}^{M_{\rm lib}}(\nnz(\bA_j)+n)\right).
\]
For a CME direction, $\nnz(\bA_j)=O(n)$, so the final term becomes
$O(N_{\rm act}n_qM_{\rm lib}n)$. With $\bA$ and its sparsity pattern fixed,
the library size appears only in this contraction term.

\subsection{Additional accuracy and timing results}
Table~\ref{tab:error_decomposition_supp} gives the numerical values used in
the Fokker--Planck refinement discussion of the main paper.
\begin{table}[tbp]
\centering
\caption{Relative coefficient-gradient error against the dense
block-Fr\'{e}chet reference~\eqref{eq:block_frechet} for the nonnormal
Fokker--Planck case at $\Delta t=0.1$.}
\label{tab:error_decomposition_supp}
\begin{tabular}{rrrr}
\toprule
$m_K$ & $n_q=8$ & $n_q=24$ & $n_q=32$ \\
\midrule
20 & $9.1\times10^{-6}$ & $1.0\times10^{-5}$ & $1.0\times10^{-5}$ \\
30 & $9.6\times10^{-6}$ & $1.1\times10^{-8}$ & $1.1\times10^{-8}$ \\
40 & $9.7\times10^{-6}$ & $1.7\times10^{-13}$ & $1.7\times10^{-13}$ \\
60 & $9.7\times10^{-6}$ & $1.8\times10^{-15}$ & $2.0\times10^{-15}$ \\
\bottomrule
\end{tabular}
\end{table}

\begin{figure}[tbp]
\centering\includegraphics[width=\linewidth]{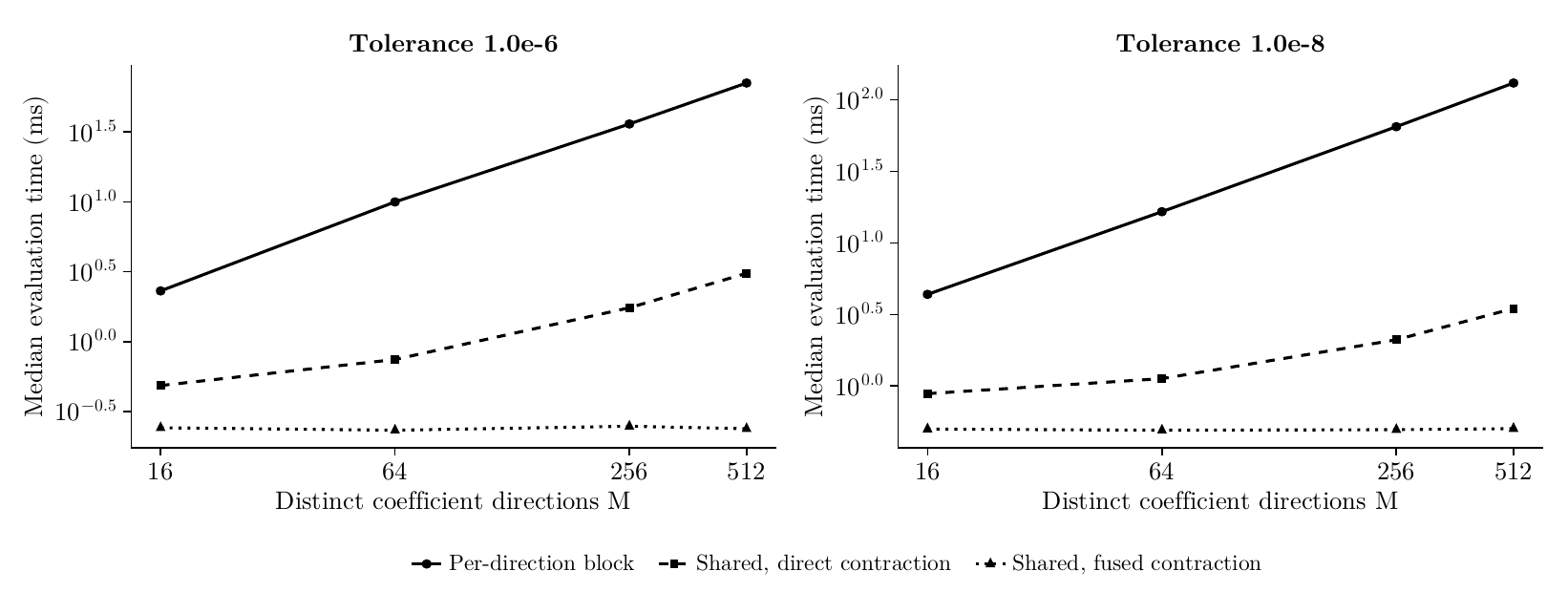}
\caption{Library scaling at matched accuracy on an Apple M1 Pro. Each time is
the median of seven evaluations after warmup and includes the forward action,
the coefficient gradient, and the input-state derivative. Lines connect the
measured configurations and are not extrapolated scaling laws.}
\label{fig:library_accuracy_supp}
\end{figure}

\begin{figure}[tbp]
\centering\includegraphics[width=\linewidth]{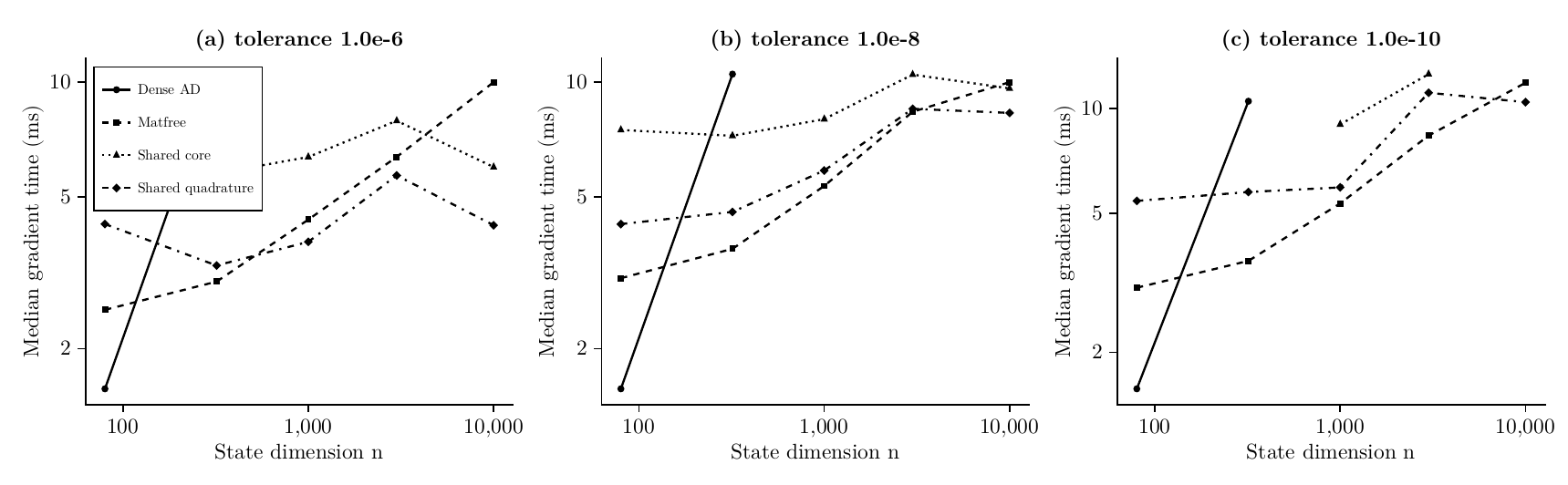}
\caption{GPU comparison within one process at fixed accuracy. Each curve uses
the fastest configuration that meets the tolerance for both the coefficient
gradient and the input-state derivative. Compilation and warmup are excluded.}
\label{fig:matfree_bw_supp}
\end{figure}

\section{CME Construction and Sparse-Inference Details}
\label{app:cme_setup}
This section describes the CME construction, initialization, and sparse
optimization used in the supporting experiments.
\subsection{Trajectory-based initialization}\label{app:initialization}

Several CME experiments use exact jump records from stochastic simulation
algorithm (SSA) paths \cite{gillespie1977exact}, rather than the snapshot-only
data used in Section~\ref{sec:experiments}. For candidate direction $j$,
let $g_j(\bx)$ be its known state-dependent propensity factor and let
$\theta_j$ be its unknown coefficient. Let $c_j^{(i)}$ be the number of times
that direction fires in path $i$, and let $X^{(i)}(t)$ be the corresponding
state. The exact-jump maximum-likelihood initializer is
\begin{equation}
\label{eq:init}
\hat{\theta}^{(0)}_j
=
\frac{\sum_{i=1}^N c_j^{(i)}}
{\sum_{i=1}^N \int_0^{T_{\mathrm{obs}}} g_j(X^{(i)}(t))\,dt}.
\end{equation}
This estimate initializes the subsequent support selection and rate
refinement. Equation~\eqref{eq:init} requires labeled reaction events so that
$c_j^{(i)}$ is known for each candidate. When several mechanisms produce the
same observed state change and their event labels are unavailable, the
individual counts cannot be separated. We then use a joint likelihood, a
nonnegative estimate from snapshot data, or a prior estimate.

\begin{remark}[Initialization from snapshot data]
With snapshot distributions but no trajectory records, a coarse nonnegative
least-squares estimate or a domain-specific prior can replace \eqref{eq:init}.
\end{remark}
The synthetic path benchmarks use the exact-jump initializer.

\subsection{Iterative thresholded optimization with sensitivity scaling}\label{app:fisher}

Starting from this estimate, the sparse inverse problem alternates smooth
local optimization with hard thresholding
\cite{brunton2016discovering,blumensath2009iterative}. To compare rates that
span several orders of magnitude, we optimize the scaled variables
$\bphi=\bD\btheta$. Here $\bD=\diag(d_j)$ is positive and $d_j$ is the norm
of column $j$ of the initial Jacobian. Thresholding therefore compares
coefficients on their sensitivity scales rather than by their raw magnitudes.
The scaled objective is
$\widetilde{\calL}(\bphi)=\calL(\bD^{-1}\bphi)$.

\begin{algorithm}[t]
\caption{Iterative Thresholded Optimization with Optional Sensitivity Scaling}
\label{alg:stls}
\begin{algorithmic}[1]
\REQUIRE Nonnegative initializer $\btheta^{(0)}$, diagonal scaling $\bD$ (identity if unused), threshold fraction $\eta_{\rm thr}\in(0,1)$, inner iteration limit $L_{\rm inner}$, outer iterations $L_{\rm outer}$
\ENSURE Sparse estimate $\hat{\btheta}$
\STATE $\bphi^{(0)}\leftarrow \bD\btheta^{(0)}$
\FOR{$\ell=0,\ldots,L_{\rm outer}-1$}
    \STATE $S_\ell \leftarrow \supp(\bphi^{(\ell)})$
    \STATE $\bphi^{(\ell+1/2)}\leftarrow \mathrm{LBFGS}(\tilde{\calL}|_{S_\ell},\nabla_{S_\ell}\tilde{\calL},\bphi^{(\ell)})$ for at most $L_{\rm inner}$ iterations
    \STATE $\bphi^{(\ell+1/2)}\leftarrow\max(\bphi^{(\ell+1/2)},0)$
    \STATE $\tau_\ell\leftarrow\eta_{\rm thr}\|\bphi^{(\ell+1/2)}\|_\infty$
    \STATE Let $\bphi^{(\ell+1)}\leftarrow\bphi^{(\ell+1/2)}$
    \STATE Let $\phi_j^{(\ell+1)}=0$ whenever $|\phi_j^{(\ell+1/2)}|<\tau_\ell$
\ENDFOR
\STATE Repeat the L-BFGS solve and nonnegative projection on the fixed support $\supp(\bphi^{(L_{\rm outer})})$
\STATE \textbf{return} $\hat{\btheta}=\bD^{-1}\hat{\bphi}$
\end{algorithmic}
\end{algorithm}

Each smooth optimization uses the adjoint gradient. Thresholding selects the
reactions, and Jacobian-column scaling makes directions with different
physical units comparable. Because the loss is nonlinear in the
coefficients, this procedure is not sequential thresholded least squares. The
experiments use the stated nonnegative projection rather than a
bound-constrained inner solve.

\begin{figure}[tbp]
\centering
\includegraphics[width=0.88\textwidth]{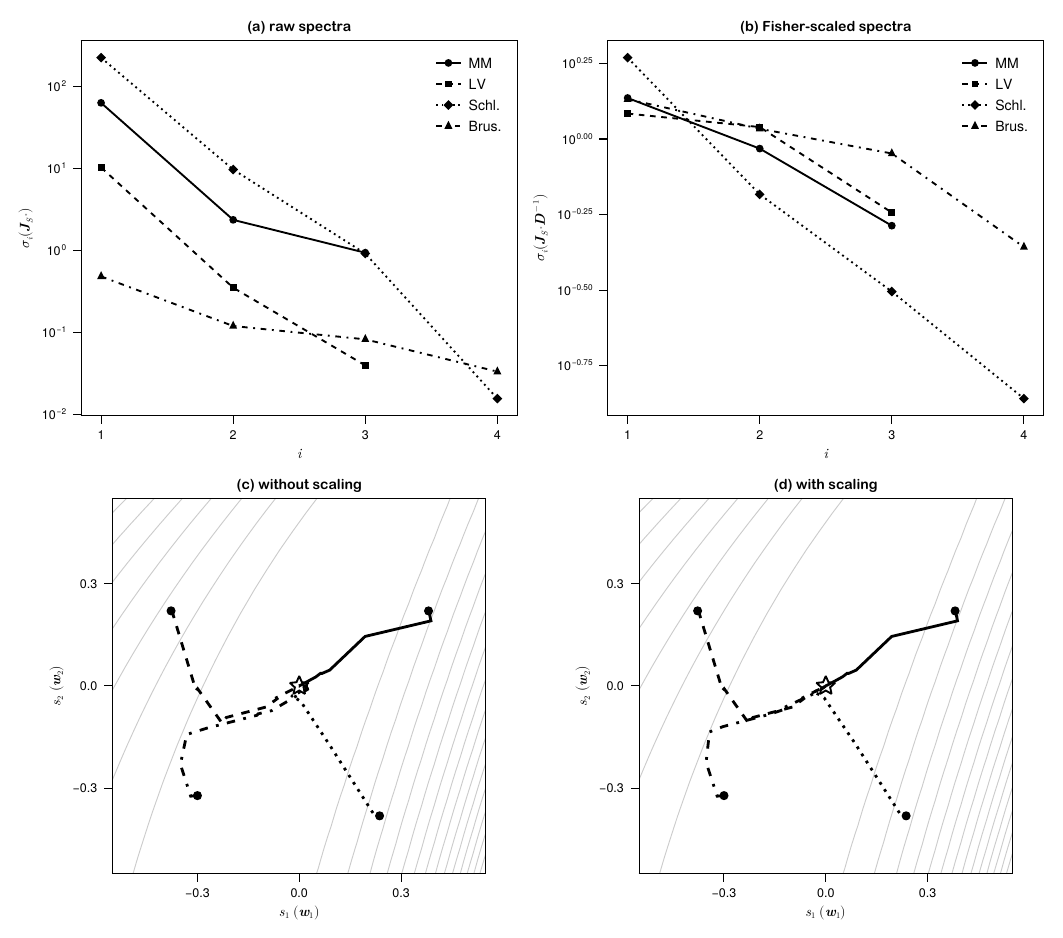}
\caption{Jacobian-column sensitivity scaling. \textbf{(a)} Singular values
$\sigma_i(\bJ_{S^*})$ for Michaelis--Menten, Lotka--Volterra, Schl\"{o}gl,
and Brusselator. The wide spread for Schl\"{o}gl reflects its
three-order-of-magnitude range of rates. \textbf{(b)} Column scaling by
$\bJ_{S^*}\bD^{-1}$ reduces the condition numbers, with the largest reduction
for Schl\"{o}gl. \textbf{(c,d)} Schl\"{o}gl loss slices along the right
singular vectors $(\bw_1,\bw_2)$ associated with the largest singular values.
The curves show optimization from four initial estimates, and the white star
marks $\btheta^*$. Without scaling~\textbf{(c)}, the iterates follow elongated
contours before converging. With scaling~\textbf{(d)}, they take shorter and
more direct paths.}
\label{fig:info_geometry}
\end{figure}

\subsection{Support stability across numerical settings and seeds}\label{app:support_stability}
This experiment separates numerical stability from variation across data
realizations. For each of six seeds, the selected support is unchanged across
two Krylov depths and four quadrature orders. Seeds 1, 4, and 6 recover the
true support in all eight settings. Seeds 2, 3, and 5 retain the same extra
direction in all eight settings. Thus the aggregate count 24/48 represents
three successful seeds repeated across eight numerical settings, not 48
independent recovery trials. Increasing the depth or quadrature order within
this grid does not remove the false positive.

We therefore distinguish stability under numerical refinement from the
effects of threshold selection, correlated library directions, and variation
in the data. Figure~\ref{fig:support_controls} places the numerical check next
to a separate perturbation of the initial estimate. This perturbation does not
change the observations. The noisy KS and finite-shot Lindblad studies in
Section~\ref{app:controls} examine variation in the observations.
\begin{figure}[tbp]
\centering\includegraphics[width=\linewidth]{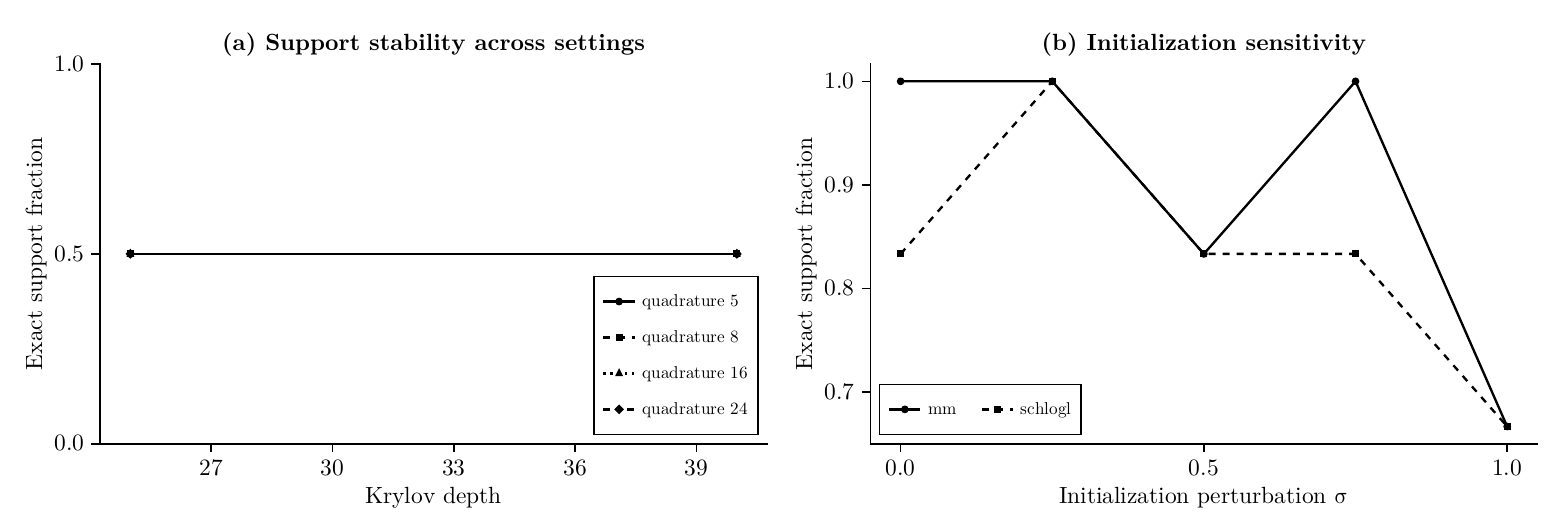}
\caption{Support stability. The left panel shows the exact-support fraction at
each Krylov depth and quadrature order. The right panel shows the exact-support
fraction under initialization perturbations for Michaelis--Menten and
Schl\"{o}gl. Fractions use all seeds at each setting.}
\label{fig:support_controls}
\end{figure}

\subsection{Finite-state CME}

Consider $d$ chemical species with copy-number state
$\bx=(x_1,\ldots,x_d)\in\N^d$. Each reaction channel $m\in[M]$ has a
stoichiometry vector $\bzeta_m\in\mathbb{Z}^d$ and a nonnegative propensity
$\alpha_m(\bx;\btheta_m)$. The propensity is the instantaneous transition
rate from $\bx$ to $\bx+\bzeta_m$. We model the system as a continuous-time
Markov chain (CTMC) on a truncated state space
$\calX=\{\bx_1,\ldots,\bx_n\}\subset\N^d$. The conservative boundary
convention in Eq.~\eqref{eq:cme_library_direction} suppresses a reaction
whose destination lies outside $\calX$. It preserves total probability within
the retained state space and differs from the leaking finite-state projection
of~\cite{munsky2006finite}.

The probability vector $\bp(t)\in\R^n$, with $p_i(t)=\Pr(X(t)=\bx_i)$,
satisfies the \emph{chemical master equation}
\begin{equation}
\label{eq:cme}
\frac{d\bp}{dt} = \bA(\btheta)\,\bp,
\end{equation}
where, with
$\alpha_m^{\calX}(\bx;\btheta_m)=
\alpha_m(\bx;\btheta_m)\mathbf 1_{\{\bx+\bzeta_m\in\calX\}}$,
the generator $\bA(\btheta)\in\R^{n\times n}$ has entries
\begin{equation}
\label{eq:generator}
[\bA(\btheta)]_{\by,\bx}
=\sum_{m=1}^M\alpha_m^{\calX}(\bx;\btheta_m)
\left(\mathbf 1_{\{\by=\bx+\bzeta_m\}}
-\mathbf 1_{\{\by=\bx\}}\right),
\qquad \bx,\by\in\calX.
\end{equation}
The columns sum to zero, so $\mathbf{1}^T\bA=\mathbf{0}$ and total probability
is conserved.
Over one observation interval $\Delta t$, the solution of~\eqref{eq:cme} is
\begin{equation}
\label{eq:discrete}
\bp(t+\Delta t) = e^{\bA(\btheta)\Delta t}\,\bp(t),
\end{equation}
so the matrix exponential
$\bP_{\Delta t}=e^{\bA\Delta t}$ is the transition operator over the interval.

\subsection{Reaction library}
\label{sec:massaction}

The reaction library has the affine form~\eqref{eq:affine}. Each
mass-action direction represents an elementary reaction with a physical rate
constant, so the recovered support identifies a reaction network. The
Fokker--Planck and Lindblad examples in Section~\ref{sec:affine} use
operator libraries constructed from their corresponding physical terms. The
gradient calculation does not depend on the choice of library.

In stochastic mass-action kinetics, a reaction consuming $\nu_i$ molecules
of species $i$ fires at a rate proportional to the number of distinct reactant
subsets. Counting ordered selections and dividing out permutations gives the
falling-factorial form
\begin{equation}
\label{eq:binom_propensity}
\phi_{\bnu}(\bx)
= \prod_{i=1}^{d}\binom{x_i}{\nu_i}
= \prod_{i=1}^{d}\frac{x_i(x_i-1)\cdots(x_i-\nu_i+1)}{\nu_i!},
\end{equation}
where $\bnu=(\nu_1,\ldots,\nu_d)\in\N^d$ is the vector of reactant stoichiometries.
For a single species, the order-$0$ term is $\binom{x}{0}=1$, which describes
a source or creation reaction independent of copy number. The order-$1$ term
is $\binom{x}{1}=x$, as in linear degradation or transport. The order-$2$
term is $\binom{x}{2}=x(x{-}1)/2$, which counts unordered pairs in a
dimerization reaction. The order-$3$ term is
$\binom{x}{3}=x(x{-}1)(x{-}2)/6$, which counts triples in a trimerization
reaction.
This convention counts unordered reactant combinations. An alternative
falling-factorial convention absorbs the factorial into the rate constant.

\subsubsection{Library parameterization.}
We express each propensity as a linear combination of $K$ basis functions of
the form~\eqref{eq:binom_propensity}:
\[
\alpha_m(\bx;\btheta_m) = \btheta_m^T\bphi(\bx), \qquad \bphi(\bx)\in\R^K,\quad \btheta_m\geq 0,
\]
where $\bphi(\bx)$ collects the combinatorial monomials up to a chosen order.
For a one-species system the cubic library is
$\bphi(x)=\bigl[1,\;x,\;\tfrac{x(x-1)}{2},\;\tfrac{x(x-1)(x-2)}{6}\bigr]^T$.
For $d$ species, we include all multivariate products
$\prod_i\binom{x_i}{\nu_i}$ with $\sum_i\nu_i\leq r$ for a chosen order $r$.
Because each propensity is linear in $\btheta_m$, the generator is linear in
the full coefficient vector $\btheta=\mathrm{vec}(\{\btheta_m\})$:
\begin{equation}
\label{eq:linear_decomp}
\bA(\btheta)=\sum_{m=1}^M \sum_{k=1}^K \theta_{m,k}\,\bA_m^{(k)},
\qquad
\theta_{m,k}\geq 0.
\end{equation}
Once the library and state space are fixed, the direction matrices
$\bA_m^{(k)}=\partial\bA/\partial\theta_{m,k}$ have entries
\begin{equation}
\label{eq:dA}
[\bA_m^{(k)}]_{\bx+\bzeta_m,\,\bx} = \phi_k(\bx),
\qquad
[\bA_m^{(k)}]_{\bx,\,\bx} = -\phi_k(\bx),
\end{equation}
when $\bx+\bzeta_m\in\calX$, with all entries in that column zero otherwise.
We compute these matrices once and reuse them throughout inference. The
synthetic libraries contain the true reactions and inactive alternatives. We
therefore seek both the active set
$S^*=\{j:\theta_j^*>0\}$ and the values of its nonzero coefficients.

\subsection{Conditional identifiability and threshold separation}
\label{sec:theory}
Let $F_S$ collect the sampled predictions, including their observation maps,
as a function of the active coefficients, and let
$J_S=DF_S(\btheta^*)$. If $J_S$ has full column rank, it contains a
nonsingular square submatrix formed from a subset of its rows. Applying the
inverse function theorem to the corresponding output coordinates shows that
two sufficiently close coefficient vectors on $S$ cannot produce the same
sampled output. Thus the coefficients are locally identifiable on $S$.
Injectivity of the full matrix-exponential Fr\'{e}chet map does not by itself
imply this rank condition. The initial states and observation maps may remove
informative directions. The diagnostic below instead assumes a fully
observed continuous path from one initial state. It complements, but does not
replace, a rank and conditioning check of the sampled Jacobian for the actual
measurements.

For threshold separation, suppose that an iterate before thresholding
satisfies $\|\widetilde{\bphi}-\bphi^*\|_\infty\leq\delta$ in the chosen
scaled coordinates.
If $\delta<\tau<\min_{j\in S^*}|\phi_j^*|-\delta$, hard thresholding
recovers $S^*$. Inactive coordinates have magnitude at most $\delta$, and
active ones have magnitude at least $\min_{S^*}|\phi_j^*|-\delta$.
This statement requires control of both the active and inactive coordinates.
A positive-definite Hessian restricted to the active coordinates controls
perturbations only within the true support and gives no bound on the inactive
coordinates. We do not claim a global convergence result for the neural
optimization.

\subsection{A computable identifiability diagnostic}
\label{sec:stanhope}

For a family that is linear in its parameters, we can also compute a
continuous-path diagnostic directly. Stanhope, Rubin, and
Swigon~\cite[Thm.~2.6]{stanhope2014identifiability} show that global
identifiability of $\bA(\btheta)=\sum_j\theta_j\bA_j$ from the single initial
condition $\bp_0$ requires
\begin{equation}
\label{eq:stanhope}
\operatorname{rank}\bG(\bA)=M,
\qquad
\bG(\bA)_{:,j}=\bigl[\bA_j\bp_0;\ \bA_j\bA\bp_0;\ \ldots;\ \bA_j\bA^{n-1}\bp_0\bigr].
\end{equation}
Column $j$ records the action of direction $j$ along the Krylov sequence of
the observed path. Computing $\bG$ requires the additional vectors
$\bA^k\bp_0$ and the products $\bA_j\bA^k\bp_0$. We therefore form it as an
offline diagnostic rather than during the gradient calculation. After column
normalization, its singular values measure the conditioning of the
identifiability test, while pairwise angles reveal directions that are hard to
separate.

The Schl\"{o}gl library illustrates both effects. All eight singular values of
the normalized $\bG$ are nonzero to working precision, so
Eq.~\eqref{eq:stanhope} has full rank. Its condition number is nevertheless
$3.6\times10^{6}$, and two pairs of columns are nearly collinear. The constant
birth feature containing $k_3$ lies $0.89^\circ$ from a death direction, while
the true cubic birth feature lies $3.66^\circ$ from an inactive cubic birth
candidate. These angles explain the behavior reported in
Section~\ref{sec:experiments}. An aggressive threshold removes the term
in the first pair, and the second pair remains difficult to distinguish after
selection. Both effects arise from the library and the observed path.
Equation~\eqref{eq:stanhope} detects them before fitting. Additional initial
conditions or observation times can separate the directions when this
ambiguity is unacceptable.

\begin{figure}[tbp]
\centering
\begin{tikzpicture}
\node[inner sep=0] (initializationplot)
  {\includegraphics[width=\linewidth]{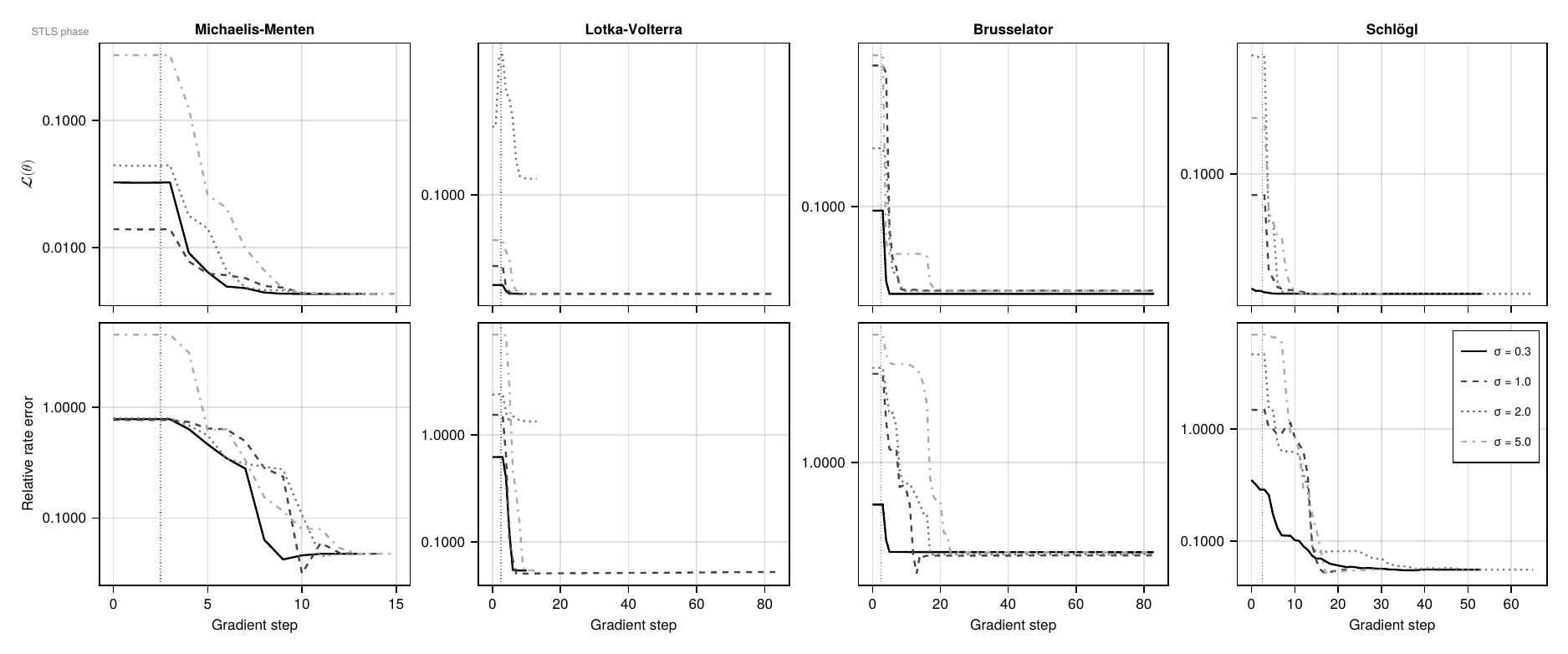}};
\fill[white] ([xshift=4pt,yshift=-2pt]initializationplot.north west)
  rectangle ++(31pt,-9pt);
\node[anchor=north west,font=\fontsize{4.5}{5.5}\selectfont,text=gray]
  at ([xshift=7pt,yshift=-5pt]initializationplot.north west)
  {Thresholding};
\end{tikzpicture}
\caption{Thresholded optimization and final refinement under perturbations
$\sigma=0.3,1,2,5$ of the initial rate estimate. Line styles and gray levels
distinguish the perturbations. The upper panels show the objective, and the
lower panels show the relative rate error. The observations are unchanged in
this experiment.}
\label{fig:initialization_paths}
\end{figure}

\section{Additional Numerical and Application Experiments}
\label{app:controls}
This section gives supporting and adverse cases from the numerical and
application studies. Each timing uses the machine and implementation on which
it was originally recorded. Comparisons across programming languages
therefore include differences in both the method and its implementation.

\begin{table}[tbp]
\centering
\caption{Kuramoto--Sivashinsky composition over five initializations. Errors
are reported as mean $\pm$ population standard deviation in percent.}
\label{tab:ks_composition_supp}
\begin{tabular}{lrr}
\toprule
Gradient calculation & Coefficient error & Prediction error \\
\midrule
Full & $1.83915\pm0.00007$ & $0.090826\pm0.000001$ \\
Input-state derivative omitted & $18.63\pm1.68$ & $1.041\pm0.093$ \\
\bottomrule
\end{tabular}
\end{table}

\begin{table}[tbp]
\centering
\caption{Lindblad identification from an overcomplete library. Each row
summarizes five independent draws of the active rates.}
\label{tab:lindblad_supp}
{\small
\begin{tabular}{rrrrrr}
\toprule
Qubits & Liouville dim. & Library size & True support & Exact & Mean rate error \\
\midrule
2 & 16   & 7  & 3 & 5/5 & $7.5\times10^{-9}$ \\
3 & 64   & 12 & 5 & 5/5 & $5.6\times10^{-9}$ \\
4 & 256  & 18 & 7 & 5/5 & $1.3\times10^{-8}$ \\
5 & 1024 & 25 & 9 & 5/5 & $1.4\times10^{-8}$ \\
6 & 4096 & 33 & 11 & 5/5 & $7.3\times10^{-9}$ \\
\bottomrule
\end{tabular}
}
\end{table}

\subsection{Extended nonlinear composition experiments}
\label{app:ks_extended}
For the split step in Eq.~\eqref{eq:ks_expanded_step}, differentiation
gives
\begin{equation}
 d\bv_k=
 \bigl(I-hcD_x\diag(\bm u_k)\bigr)d\bm u_k
 -\frac{h}{2}D_x(\bm u_k\odot\bm u_k)\,dc.
\end{equation}
For a loss $\ell(\bm u_{k+1})$, the reverse derivatives are
\begin{align}
 \nabla_{\bv_k}\ell
 &=e^{h\bA(a,b)^T}\nabla_{\bm u_{k+1}}\ell,\\
 \frac{\partial\ell}{\partial c}
 &=-\frac{h}{2}(\nabla_{\bv_k}\ell)^T D_x(\bm u_k\odot\bm u_k),\\
 \nabla_{\bm u_k}\ell
 &=\bigl(I-hcD_x\diag(\bm u_k)\bigr)^T\nabla_{\bv_k}\ell.
\end{align}
The derivatives with respect to $a$ and $b$ use the forward and adjoint paths
\[
 \bm f(s)=e^{sh\bA}\bv_k,\qquad
 \bm\lambda(s)=e^{(1-s)h\bA^T}\nabla_{\bm u_{k+1}}\ell,
\]
and the directions $-D_{xx}$ and $-D_{xxxx}$ in the main adjoint formula.

\begin{figure}[tbp]
\centering\includegraphics[width=\linewidth]{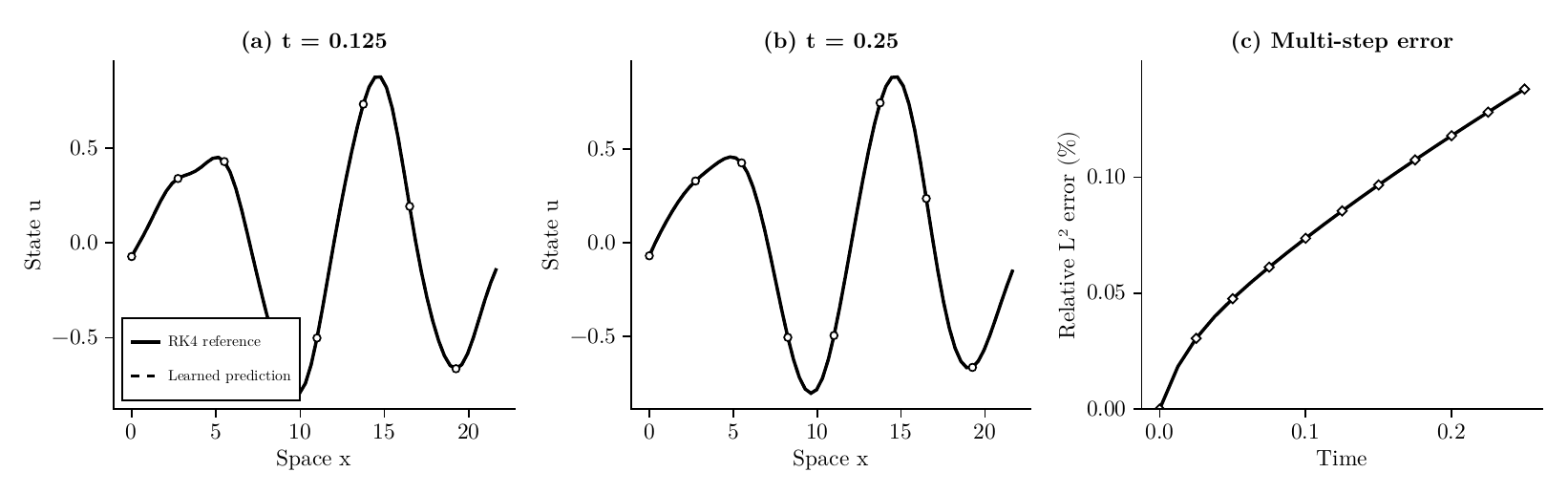}
\caption{Kuramoto--Sivashinsky test-path prediction using the complete
derivative. Solid curves show the RK4 reference. Dashed curves and open
markers show the inferred model at two times, and the right panel gives the
relative $L^2$ error. These short paths test composition of the derivatives,
not long-time chaotic prediction.}
\label{fig:ks_trajectory}
\end{figure}

\paragraph{Sparse nonlinear library.}
In a second KS experiment, the library contains 15 candidates. Its linear directions are
$-D_x$, $-D_{xx}$, $D_xD_{xx}$ and $-D_{xxxx}$. The remaining 11 candidates
are $-D_x\bm q_r(\bm u)/r$, with $r=2,\ldots,12$, where
$[\bm q_r(\bm u)]_i=u_i^r$.
The true support contains only $-D_{xx},-D_{xxxx}$ and the quadratic flux.
The training data contain 20 transitions from each of three fields over
$T=0.25$. Independent Gaussian noise ranges from zero to $1\%$ of the
root-mean-square amplitude of each snapshot. A fourth noise-free path tests
prediction through $T=0.5$. We optimize an $\ell_1$-regularized fit with
L-BFGS, apply a threshold, and refit the selected terms without a penalty.
The penalty $10^{-7}$ and threshold $0.08$ are calibrated on seed 11 at
$0.5\%$ noise and then fixed for 20 initializations at every noise level. All
80 runs recover exactly the three active terms with no false positives.
Table~\ref{tab:ks_noise_sweep} reports this support result separately from the
coefficient and prediction errors. Refitting reduces both errors through
$0.5\%$ noise but increases them at $1\%$. This experiment tests sparse
semilinear identification on a small grid and over short time intervals using
the complete derivative.

\begin{table}[tbp]
\centering
\caption{KS robustness over 20 initializations per noise level. Coefficient
and noise-free test-path errors are percentages. D denotes the initial
sparse fit and R denotes refitting on the selected support. Both stages retain
the support selected by the initial fit.}
\label{tab:ks_noise_sweep}
\begin{tabular}{rrrrrr}
\toprule
Noise & Exact & Coeff. D & Coeff. R & Test D & Test R \\
\midrule
$0\%$    & 20/20 & 5.48 & 1.72 & 0.513 & 0.140 \\
$0.25\%$ & 20/20 & 5.07 & 1.98 & 0.480 & 0.157 \\
$0.5\%$  & 20/20 & 4.22 & 3.21 & 0.407 & 0.255 \\
$1\%$    & 20/20 & 5.98 & 8.79 & 0.508 & 0.713 \\
\bottomrule
\end{tabular}
\end{table}

\subsection{Context-dependent support selection}
\label{sec:context_experiment}

To separate the rate parameterization from prior knowledge of the active
support, we consider a Brusselator on a $13\times13$ grid ($n=169$) whose
rates depend on context. The library contains four candidate stoichiometries
and ten falling-factorial propensity features for each stoichiometry, giving
40 directions. A nonnegative matrix $W\in\R^{40\times3}$ defines
$\theta_j(a,b)=W_{j,:}[1,a,b]^T$. Unlike selection by a separate simplex within
each stoichiometric channel, this parameterization imposes neither a
one-feature-per-channel constraint nor a prescribed support size. The true
rates depend linearly on $(a,b)$, so a linear rate map is sufficient.

Training uses exact snapshots generated by the model at four contexts. For
each context, the data include six initial states and observation intervals
$0.05$, $0.20$, and $0.50$. We first compute an $\ell_1$-regularized fit,
threshold its rates, and refit the selected support without a penalty. The
threshold $0.20$ is calibrated on seed 11 and fixed before running the
confirmation seeds 22, 33, 44, and 55. All five runs recover exactly the four
active directions with no false positives. This result includes the four
untouched confirmation runs. Across all five runs, the mean relative
coefficient error is $0.875\%$, the prediction error at four test contexts is
$1.11\%$, and the error at the longer horizon $t=1.5$ is $1.05\%$.

This experiment isolates context-dependent mechanism selection without
channel-wise constraints. The data are noiseless, the state space is small,
and the true rate map is linear. A nonlinear rate map requires separate tests
on unseen contexts under the same support-selection protocol.

\subsection{Under-resolved training and independent adjoint resolution}
\label{app:gradient_stress}
The choice of derivative matters when the forward approximation is coarse.
In an 80-state, ten-direction Fokker--Planck experiment, we hold fixed the
finite-Arnoldi forward map, exact targets, four initial states, parameter
initialization, and Adam updates. With $h=1$, $m_K=30$, and $n_q=8$, the
adjoint calculation gives an exact-exponential prediction error of $19.3\%$
and a coefficient error of $374\%$ after 2000 epochs. Differentiating the
finite Arnoldi calculation gives errors of $4.58\%$ and $55.2\%$. This
single-seed comparison shows that optimization of an under-resolved
finite-Arnoldi model can favor the derivative of the finite calculation. The
forward and adjoint approximations must both be resolved when the target is
the gradient of the exact exponential action.

At a stationary initial distribution, a nearly exact forward value can
coexist with an inaccurate coefficient gradient. In the recorded CTMC
experiment, the forward relative error is $1.86\times10^{-11}$ at depth three,
while the coefficient-gradient error is $0.180$. Increasing the adjoint depth
and checking its residual resolves the latter error. A stopping rule based
only on the forward calculation cannot detect it. The stationary-start,
invariant, and precision checks below give the corresponding numerical
details.

\subsection{Large sparse Lindblad contraction evaluation}
\label{app:lindblad_kernel}
The explicit sparse benchmark reaches eleven qubits on a 16~GB Apple M1 Pro.
Its generator contains local fields, nearest-neighbor exchange, dephasing,
and decay. The operator library contains the local dephasing and
nearest-neighbor exchange directions, for a total of $2q-1$ directions. This
benchmark measures gradient computation for a fixed library rather than
support recovery.

Table~\ref{tab:lindblad_kernel_scaling} reports median times after warmup. We
use three timed runs through ten qubits and two runs at eleven qubits. The
first four rows use Krylov depth 40 and eight quadrature nodes. The
eleven-qubit calculation uses depth 24 and six nodes. Increasing these values
to 32 and 10 changes its gradient by $2.56\times10^{-13}$. Workspace memory
includes the two Arnoldi bases and the fused sparse-direction cache, while
operator storage includes the generator and every library direction. The
eleven-qubit process reaches a maximum resident memory of 7.31~GB without
swapping.

\begin{table}[tbp]
\centering\small
\caption{Explicit sparse Lindblad scaling on a 16~GB Apple M1 Pro. Times are
in milliseconds. The final column gives the relative change in the gradient
after increasing the depth and quadrature order as described in the text.}
\label{tab:lindblad_kernel_scaling}
\resizebox{\linewidth}{!}{%
\begin{tabular}{rrrrrrrrr}
\toprule
Qubits & Dimension & Directions & $m_K/n_q$ & Operator MiB & Workspace MiB & Forward & Gradient & Relative change \\
\midrule
7  & 16,384    & 13 & 40/8 & 8.9   & 28.5   & 17.5  & 188.7 & $2.47\times10^{-15}$ \\
8  & 65,536    & 15 & 40/8 & 40.6  & 114.6  & 75.3  & 532.7 & $2.61\times10^{-15}$ \\
9  & 262,144   & 17 & 40/8 & 183.0 & 476.9  & 310.2 & 2,757.6 & $4.02\times10^{-15}$ \\
10 & 1,048,576 & 19 & 40/8 & 814.0 & 1,994.3 & 1,471.1 & 9,346.6 & $7.68\times10^{-15}$ \\
11 & 4,194,304 & 21 & 24/6 & 3,584.0 & 6,248.5 & 3,382.9 & 13,341.5 & $2.56\times10^{-13}$ \\
\bottomrule
\end{tabular}
}
\end{table}

Quadrature also permits the forward and adjoint Arnoldi calculations to use
the same workspace. We test this memory schedule on the local Lindblad
generator
\[
\mathcal L(\rho)=-\mathrm{i}\left[
  \sum_{s=1}^q\omega_s Z_s,\rho\right]
 +\sum_{s=1}^q\gamma_s(Z_s\rho Z_s-\rho),
\]
with positive, site-dependent coefficients. The generator is diagonal in the
local matrix-unit basis. Its action therefore requires $O(4^q)$ storage rather
than an explicit $4^q\times4^q$ matrix. We differentiate a local
transverse-field commutator and a local decay dissipator, neither of which
commutes with the generator. The initial density matrix and the derivative of
the loss are dense product operators.

We use $h=2$, Krylov depth 32, and eight Gauss--Legendre nodes. Exact divided
differences of the diagonal generator give an independent reference gradient.
At depth 24, the relative errors are $2.59\times10^{-11}$ for ten qubits and
$1.95\times10^{-10}$ for eleven qubits. Depth 32 reduces both errors below
$10^{-11}$ and is used throughout
Table~\ref{tab:lindblad_streaming_memory}.

The quadrature calculation stores the forward values at the eight nodes and
then overwrites the forward Arnoldi basis with the adjoint basis. The
retained-basis core constructs both bases once and contracts its singular
factors one at a time. The recomputed-basis core stores the selected adjoint
factors, overwrites the adjoint basis, and constructs the forward basis a
second time.

\begin{table}[tbp]
\centering\small
\caption{Memory schedules for the matrix-free Lindblad gradient on a 16~GB
Apple M1 Pro. The dimensions are $4^q$. We do not run the retained-basis core
at twelve qubits because its two depth-32 bases alone require 16~GiB. Errors
are measured against the exact divided-difference gradient. Each measurement
starts in a fresh process. Times cover the gradient calculation, and
\texttt{/usr/bin/time -l} reports the peak memory footprint.}
\label{tab:lindblad_streaming_memory}
\resizebox{\linewidth}{!}{%
\begin{tabular}{rrlrrrrr}
\toprule
Qubits & Dimension & Evaluation & Arnoldi passes & Rank & Time (s) & Peak (GiB) & Relative error \\
\midrule
10 & 1,048,576  & Quadrature       & 2 & -- & 5.30   & 1.26  & $1.52\times10^{-13}$ \\
   &            & Core             & 2 & 10 & 5.50   & 1.66  & $5.89\times10^{-14}$ \\
   &            & Core, recomputed & 3 & 10 & 8.19   & 1.26  & $5.90\times10^{-14}$ \\
11 & 4,194,304  & Quadrature       & 2 & -- & 11.54  & 3.44  & $2.20\times10^{-13}$ \\
   &            & Core             & 2 & 10 & 12.69  & 5.03  & $2.15\times10^{-13}$ \\
   &            & Core, recomputed & 3 & 10 & 17.82  & 3.40  & $2.15\times10^{-13}$ \\
12 & 16,777,216 & Quadrature       & 2 & -- & 94.06  & 12.30 & $4.49\times10^{-14}$ \\
   &            & Core             & 2 & -- & --     & $>16$ & -- \\
   &            & Core, recomputed & 3 & 10 & 124.20 & 12.53 & $1.65\times10^{-12}$ \\
\bottomrule
\end{tabular}
}
\end{table}

\subsection{CPU/GPU timing comparisons and dense differentiation}
In the portable sweep, the local transition rates remain fixed as the state
dimension increases. For each method, we choose the fastest tested
configuration that meets the error tolerance for both the coefficient
gradient and the input-state derivative. Table~\ref{tab:matfree}
compares GPU quadrature with differentiated Arnoldi. The table below adds CPU
measurements from the same JAX process. We do not interpolate between tested
configurations, and a dash indicates that none meets the tolerance.
Figure~\ref{fig:dense_bw} compares dense JAX and PyTorch differentiation with
the sparse Julia calculation. JAX uses compilation, PyTorch runs eagerly, and
the Julia implementation uses sparse operators.

\begin{table}[tbp]
\centering\small
\caption{Same-process Apple M1 CPU timings at tolerance $10^{-6}$, in
milliseconds.}
\begin{tabular}{llllll}
\toprule
CPU & $n$ & Matfree & Quadrature & Projection & Dense AD \\
\midrule
Apple M1 & 80 & 0.717 & 2.96 & 0.886 & 1.2 \\
Apple M1 & 320 & 1.87 & 3.83 & 1.76 & 15.1 \\
Apple M1 & 1000 & 6.37 & 7.8 & 5.56 & N/A \\
\bottomrule
\end{tabular}
\end{table}
\begin{figure}[tbp]
\centering\includegraphics[width=\linewidth]{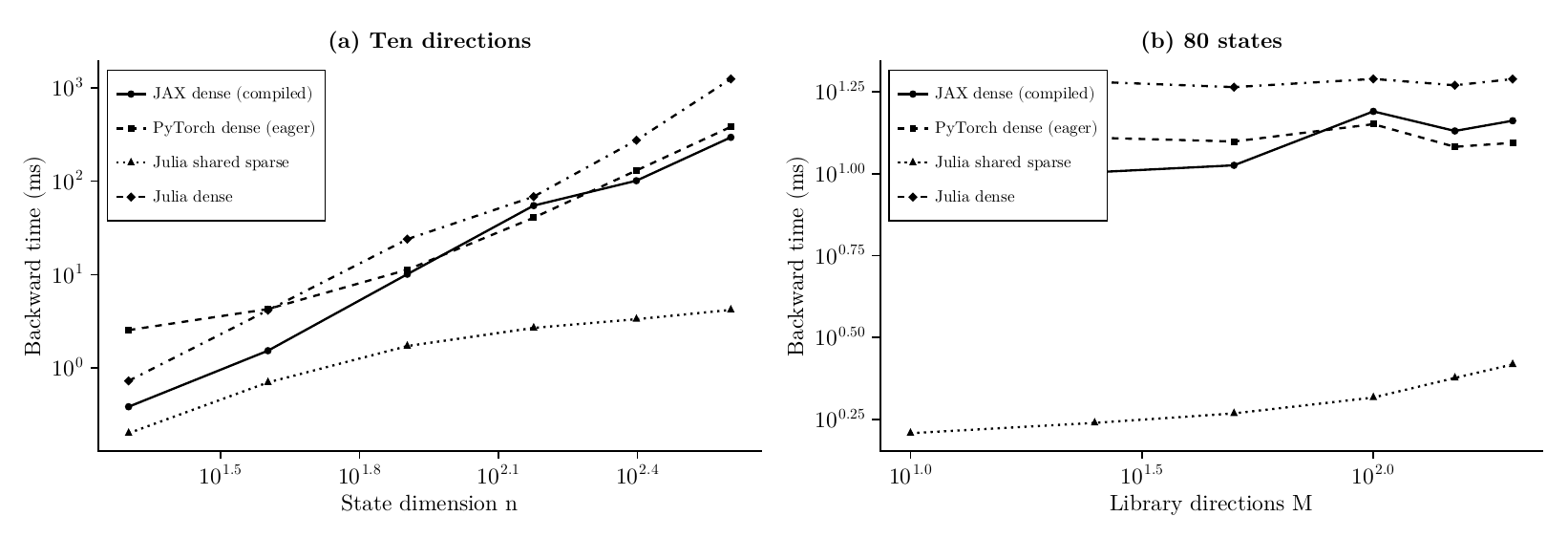}
\caption{Dense automatic differentiation and sparse Julia timing sweeps. The
left panel varies the state dimension with ten directions, and the right panel
varies the library size at $n=80$. Each curve retains its original language
and execution mode. This figure is separate from the same-backend comparison
with differentiated Arnoldi.}
\label{fig:dense_bw}
\end{figure}

\subsection{Two-sided projected factor, shared sparsity and precision}
\label{sec:core}
Quadrature can be replaced by a small block exponential that uses the same
two bases. Let $\bv=\beta_f V_f\bdelta_1$ and
$\nabla_{\bm y}\ell=\beta_a V_a\bdelta_1$. The block exponential
\begin{equation}
\label{eq:core}
C_m=\left[\exp\!\left(h\begin{pmatrix}
H_a&\beta_a\beta_f\bdelta_1\bdelta_1^*\\0&H_f^*
\end{pmatrix}\right)\right]_{12}.
\end{equation}
contains the two-sided projected Fr\'{e}chet factor. If
$C_m\approx U_r\Sigma_r W_r^*$, the contractions become
$\partial\ell/\partial\theta_j\approx\operatorname{Re}\sum_{k=1}^r
\sigma_k(V_a u_k)^*\bA_j(V_f w_k)$.
This is the two-sided approximation of
Kressner~\cite{kressner2019bivariate} and Algorithm~2 of Kandolf et
al.~\cite{kandolf2021lowrank}, specialized to the rank-one direction obtained
from the derivative of the loss. Its factors are shared across the affine
library. For orthonormal bases, Cauchy--Schwarz gives
\[
\left|\delta\left(\frac{\partial\ell}{\partial\theta_j}\right)\right|
\leq\|C_m-U_r\Sigma_r W_r^*\|_F\|\bA_j\|_F.
\]
This bound controls truncation of the projected factor, but not the Krylov
projection error. The block exponential evaluates the projected integral
exactly and therefore has no quadrature remainder.

When several directions share matrix entries, we store their values in a
sparse $|P|\times M_{\rm lib}$ matrix over the union pattern
$P=\bigcup_j\supp(\bA_j)$. The low-rank product is evaluated once at each
entry of this union, after which multiplication by the sparse value matrix
gives every directional contraction. Cache construction and storage are
counted separately. Compression and fused contraction need not reduce the
cost for a small library.

We measure factor compression and the union-pattern contraction separately.
Table~\ref{tab:union_pattern} fixes the state size, bases, and generator and
varies only the library size. The relative difference compares two
implementations of the same projected factor. It is not an error against the
exact exponential derivative.

\begin{table}[tbp]
\centering\small
\caption{Union-pattern contraction at $n=400$, depth 40, and $h=0.4$.
Timings exclude construction of the sparse cache.}
\label{tab:union_pattern}
\begin{tabular}{lllll}
\toprule
Directions & Direct (ms) & Fused (ms) & Speedup & Relative difference \\
\midrule
10 & 0.488 & 0.448 & 1.09 & $1.14\times10^{-15}$ \\
50 & 0.688 & 0.496 & 1.38 & $1.29\times10^{-15}$ \\
200 & 1.43 & 0.648 & 2.21 & $1.89\times10^{-15}$ \\
1000 & 5.57 & 1.52 & 3.67 & $2.25\times10^{-15}$ \\
\bottomrule
\end{tabular}
\end{table}
In a wide-stencil check with $n=400$, bandwidth parameter 25, and forward
depth 60, the coefficient-gradient errors are $9.71\times10^{-15}$ for
differentiated Arnoldi and $9.19\times10^{-12}$ for the rank-eight symmetric
factor. Reducing only the adjoint depth to 20 increases the latter error to
$5.27\times10^{-9}$, which shows the effect of resolving the two sides
independently. A second comparison uses a dense reference with $n=80$, 50
directions, and depth 40. The projected calculation takes 0.993 milliseconds,
and differentiated Arnoldi takes 0.764 milliseconds. Their coefficient errors
are $6.25\times10^{-11}$ and $3.46\times10^{-9}$, respectively.

\begin{table}[tbp]
\centering\small
\caption{Precision check against a Float64 dense reference. The finite-depth
calculations in this table are not fully converged.}
\begin{tabular}{llllll}
\toprule
System & $h$ & Precision & Forward error & Parameter error & State error \\
\midrule
Fokker--Planck & 0.3 & float64 & $2.05\times10^{-6}$ & $9.50\times10^{-5}$ & $1.14\times10^{-4}$ \\
Fokker--Planck & 0.3 & float32 & $2.52\times10^{-6}$ & $1.34\times10^{-4}$ & $2.42\times10^{-4}$ \\
Fokker--Planck & 1.0 & float64 & $1.12\times10^{-5}$ & $1.95\times10^{-4}$ & $2.79\times10^{-4}$ \\
Fokker--Planck & 1.0 & float32 & $1.72\times10^{-5}$ & $6.17\times10^{-4}$ & $1.51\times10^{-3}$ \\
Lindblad & 1.0 & float64 & $2.72\times10^{-16}$ & $4.56\times10^{-10}$ & $1.50\times10^{-15}$ \\
Lindblad & 1.0 & float32 & $2.74\times10^{-7}$ & $6.60\times10^{-7}$ & $1.06\times10^{-6}$ \\
\bottomrule
\end{tabular}
\end{table}
\subsection{Stationary starts, invariants and residual diagnostics}
At an exactly stationary CTMC input, the forward Arnoldi process may terminate
after one vector because $\bA\bp_0$ is a multiple of $\bp_0$. The tested
differentiated-Arnoldi implementation returns nonfinite gradients at this
breakdown. Replacing the breakdown with the terminated forward value instead
produces a zero coefficient derivative and misses the sensitivity of the
exact exponential action. Near stationarity, differentiation of the finite
Arnoldi operations is well defined, but its result can still differ
substantially from the derivative of the exact action. The separate adjoint
space resolves the path that remains nontrivial.

\begin{table}[tbp]
\centering\small
\caption{Near-stationary CTMC derivative check at perturbation $10^{-10}$.
Gradient errors use a dense evaluation of the exact exponential derivative as
the reference.}
\begin{tabular}{llll}
\toprule
Depth & Forward error & Adjoint Krylov gradient & Matfree gradient \\
\midrule
3 & $1.14\times10^{-11}$ & 0.125 & 0.831 \\
4 & $3.66\times10^{-12}$ & 0.0456 & 0.713 \\
8 & $9.65\times10^{-14}$ & $1.01\times10^{-3}$ & 0.0515 \\
12 & $1.11\times10^{-15}$ & $3.99\times10^{-6}$ & $3.63\times10^{-4}$ \\
\bottomrule
\end{tabular}
\end{table}
For conservative directions and an invariant adjoint vector $\bm c$, we have
$\bA^*\bm c=\bA_j^*\bm c=0$, so the exact parameter gradient vanishes
and the input-state gradient equals $\bm c$. The CTMC and Lindblad checks,
using an equivalent real representation in the second case, confirm both
identities for the adjoint Krylov calculation. Explicit projection onto the
invariant restores them for differentiated Arnoldi. Preservation of these
invariants does not certify the remaining gradient components.

For the CTMC diagnostic, we compute the residual bound from
Corollary~\ref{cor:ctmc_residual_bound} for each coefficient. We add the
absolute difference between the $n_q$- and $2n_q$-point quadrature results as
an empirical estimate of the quadrature error. In every saved test, the sum
exceeds the error measured against the dense reference.

We use this sum as the half-width of an interval around the magnitude of each
approximate gradient. The ranking is resolved only if the lower endpoint for
every selected direction exceeds the upper endpoint for every unselected
direction. Overlap makes the test inconclusive even when the computed ranking
is correct. The residual contribution is bounded, while the difference
between the two quadrature rules is an empirical estimate rather than a
rigorous remainder bound.
\subsection{Library size, correlation and misspecification}
In a noiseless 30-state CTMC experiment, the library contains all 870 directed
edges. Four edges are active, and three training distributions place mass at
their source states. An $\ell_1$-regularized fit is followed by a fixed
threshold and an unpenalized refit. All five seeds recover the four edges
without knowing the support size. The mean relative errors in the
coefficients, test transitions, and long-horizon predictions are
$2.57\times10^{-7}$, $9.58\times10^{-8}$, and $9.98\times10^{-8}$.

A large library alone is less problematic than a library with nearly
collinear candidates. In the correlation sweep below, accurate prediction can
coexist with an incorrect support. At correlation 0.9 and penalty
$5\times10^{-5}$, exact recovery changes from 0/5 after Adam to 5/5 after
L-BFGS refinement. Optimization and threshold separation therefore also
affect support recovery. In the misspecified libraries, the planted mechanism
is absent, so an exact-support fraction is not defined.

\begin{table}[tbp]
\centering\small
\caption{Five-seed library experiments. Setting denotes correlation or
misspecification magnitude. It is zero for the size sweep.}
\begin{tabular}{llllll}
\toprule
Sweep & Library size & Setting & Exact fraction & Test error & Long-time error \\
\midrule
Size & 10 & 0 & 1 & $1.12\times10^{-8}$ & $1.16\times10^{-8}$ \\
Size & 20 & 0 & 1 & $1.09\times10^{-8}$ & $1.08\times10^{-8}$ \\
Size & 40 & 0 & 1 & $7.99\times10^{-9}$ & $8.32\times10^{-9}$ \\
Size & 80 & 0 & 1 & $7.68\times10^{-9}$ & $7.93\times10^{-9}$ \\
Correlation & 40 & 0.9 & 0 & 0.019 & 0.0551 \\
Correlation & 40 & 0.99 & 0 & $8.26\times10^{-3}$ & 0.0257 \\
Correlation & 40 & 0.999 & 0 & $3.27\times10^{-3}$ & 0.0111 \\
Misspecified & 40 & 0.01 & N/A & $1.72\times10^{-4}$ & $8.14\times10^{-4}$ \\
Misspecified & 40 & 0.05 & N/A & $8.54\times10^{-4}$ & $4.03\times10^{-3}$ \\
Misspecified & 40 & 0.1 & N/A & $1.70\times10^{-3}$ & $7.96\times10^{-3}$ \\
Misspecified & 40 & 0.2 & N/A & $3.39\times10^{-3}$ & 0.0156 \\
\bottomrule
\end{tabular}
\end{table}
\subsection{Neural application comparisons}
These discovery comparisons use the same library but different integrators,
observation intervals, and optimization budgets. The exponential fit at the
shorter interval also misses one true term. A finer time discretization alone
does not guarantee support recovery at a fixed training budget.

\begin{table}[tbp]
\centering\small
\caption{Stuart--Landau integrator comparisons during the discovery stage.
The true model contains eight terms.}
\begin{tabular}{lllll}
\toprule
Method & True hits & False hits & Epochs & Seconds \\
\midrule
Exponential h=1.0 & 8 & 0 & 800 & 91.6 \\
Euler h=1.0 & 3 & 9 & 800 & 37.9 \\
RK4$\times$1 h=1.0 & 2 & 8 & 800 & 97.6 \\
RK4$\times$32 h=1.0 & 1 & 1 & 100 & 294 \\
Exponential h=0.1 & 7 & 1 & 800 & 74.2 \\
\bottomrule
\end{tabular}
\end{table}

\begin{table}[tbp]
\centering\small
\caption{Lindblad observation studies. Entries are mean componentwise rate
errors on test contexts, in percent, over three seeds.}
\begin{tabular}{lllll}
\toprule
Qubits & Full state & Local Pauli & Single-site Pauli & Local $Z$ \\
\midrule
2 & 5.94 & 5.44 & 6.6 & 338 \\
3 & 6.97 & 7.93 & 15.5 & 293 \\
4 & 5.8 & 6.12 & 6.96 & 295 \\
\bottomrule
\end{tabular}
\end{table}

\begin{table}[tbp]
\centering\small
\caption{Lindblad experiments with finite-shot local Pauli observations. Each
entry is the mean rate error in percent over three seeds.}
\begin{tabular}{lllll}
\toprule
Qubits & Noiseless & 100 shots & 1,000 shots & 10,000 shots \\
\midrule
2 & 5.44 & 11.5 & 6.56 & 5.56 \\
3 & 7.93 & 16.3 & 8.69 & 7.92 \\
4 & 6.12 & 15.9 & 7.45 & 6.2 \\
\bottomrule
\end{tabular}
\end{table}
\subsection{Additional path and diffusion experiments}
The exponential derivative also applies to path likelihoods. For an observed
transition $i\to j$ with positive probability, the loss is
$-\log([e^{hA}]_{ji})$. Its derivative at the right endpoint differs from the
derivative of a squared snapshot error. The archived Ornstein--Uhlenbeck,
double-well, and two-dimensional diffusion studies use this likelihood. The
table reports their test coefficient errors separately from the matched
gradient comparisons. A missing Euler--Maruyama entry denotes a nonfinite
result.

\begin{table}[tbp]
\centering\small
\caption{Archived test coefficient errors, in percent. For two-dimensional
entries, we average the relative norm of the two-coefficient vector over test
contexts. This avoids division by a vanishing individual coefficient.}
\begin{tabular}{lll}
\toprule
Experiment / loss & Contexts & Mean error (\%) \\
\midrule
Trajectory OU (likelihood) & 5 & 1.94 \\
Trajectory OU (snapshot MSE) & 5 & 1.2 \\
Double well (exponential) & 6 & 6.11 \\
Double well (Euler--Maruyama) & 6 & N/A \\
2D trajectory (likelihood) & 5 & 79.4 \\
2D trajectory (snapshot MSE) & 5 & 58.6 \\
\bottomrule
\end{tabular}
\end{table}

\FloatBarrier
\fi
\ifmainpaper\clearpage\fi
\bibliographystyle{siamplain}
\bibliography{references_cme}
\end{document}